\documentclass{article}
\usepackage{arxiv}

\usepackage[utf8x]{inputenc}
\usepackage{algorithmic}
\usepackage{amsmath,amssymb,amsfonts,amsthm}
\usepackage{balance}
\usepackage{color}
\usepackage{enumerate}
\usepackage{graphicx}
\usepackage{subcaption}
\usepackage{textcomp}
\usepackage{mathtools}
\usepackage{url}
\usepackage{xcolor}
\usepackage{wrapfig}
\usepackage{bm}
\usepackage[normalem]{ulem}
\usepackage{listings}
\usepackage{mwe,tikz}
\usepackage[percent]{overpic}
\usepackage{enumitem}

\newcommand{\srv}{\mathbf{s}}
\newcommand{\Pt}{\mathbf{P}}

\newcommand{\rate}{\bm{\mu}}
\newcommand{\qlen}[1]{\tilde{\mathbf{x}}_{#1}}
\newcommand{\qlens}[1]{\hat{\mathbf{x}}_{#1}}
\newcommand{\use}[1]{\hat{\mathbf{u}}_{#1}}
\newcommand{\qlenFnc}[1]{\mathbf{x}\left( #1 \right) }
\newcommand{\qlenFncS}[2]{\mathbf{x}_{#2}\left( #1 \right) }
\newcommand{\useFnc}[1]{\mathbf{u}\left( #1 \right) }

\newcommand{\id}{\mathbf{I}}
\newcommand{\node}[1]{\textbf{#1}}

\renewcommand{\vec}[1]{\ensuremath{\mathbf{#1}}}

\newtheorem{thm}{Theorem}[section]
\newtheorem*{thm*}{Theorem}

\providecommand{\acks}[1]{\parbox{\textwidth}{\textbf{\textit{Acknowledgements: }} #1}}
\newtheorem{example}{Example}

\begin{document}

\title{Learning Queuing Networks by\\Recurrent Neural Networks}
\chead{Learning Queuing Networks by Recurrent Neural Networks}
\def\keywordname{{\bfseries \emph{Keywords}}}%

\renewcommand{\headeright}{}
\renewcommand{\undertitle}{}
\author{
Giulio Garbi\\
IMT School for Advanced Studies Lucca\\
Piazza San Francesco, 19, \\
55100 Lucca, Italy\\
\texttt{giulio.garbi@imtlucca.it}
\And
Emilio Incerto\\
IMT School for Advanced Studies Lucca\\
Piazza San Francesco, 19, \\
55100 Lucca, Italy\\
\texttt{emilio.incerto@imtlucca.it}
\And
Mirco Tribastone\\
IMT School for Advanced Studies Lucca\\
Piazza San Francesco, 19, \\
55100 Lucca, Italy\\
\texttt{mirco.tribastone@imtlucca.it}
}

\maketitle

\begin{abstract}
It is well known that building analytical performance models in practice is difficult because it requires a considerable degree of proficiency in the underlying mathematics. In this paper, we propose a machine-learning approach to derive performance models from data. We focus on queuing networks, and crucially exploit a deterministic approximation of their average dynamics in terms of a compact system of ordinary differential equations. We encode these equations into a recurrent neural network whose weights can be directly related to model parameters. This allows for an interpretable structure of the neural network, which can be trained from system measurements to yield a white-box parameterized model that can be used for prediction purposes such as what-if analyses and capacity planning. Using synthetic models as well as a real case study of a load-balancing system, we show the effectiveness of our technique in yielding models with high predictive power.             
\end{abstract}



\section{Introduction}

\paragraph*{Motivation} Performance metrics such as throughput and response time are important factors that impact on the quality of a software system as perceived by users. They indicate how well the software behaves, thus complementing functional properties that concern what the software does. 
A traditional way of reasoning about the performance in a software system is by means of profiling. A tool such as \textsl{Gprof} executes the program and allows the identification of the program locations that are most performance sensitive~\cite{gprof}. The main limitation is that this information is valid for the specific run with which the program is exercised; different inputs lead to different performance profiles in general. Thus, while profiling can detect the presence of performance anomalies, it lacks generalizing and predictive power (see also~\cite{Zaparanuks:2012:AP:2254064.2254074}). 

As with all scientific and engineering disciplines, predictions can be made with models. Software performance models are mathematical abstractions whose analysis provides quantitative insights into real systems under consideration~\cite{DBLP:books/daglib/0027475}. 
Typically, these are stochastic models based on Markov chains and other higher-level formalisms such as queueing networks, stochastic process algebra, and stochastic Petri nets (see, e.g.,~\cite{DBLP:books/daglib/0027475} for a detailed account). Although they have proved effective in describing and predicting the performance behavior of complex software systems (e.g.,~\cite{bolch,DBLP:conf/sfm/Stewart07}),  a pressing limitation is that the current state of the art hinges on considerable craftsmanship to distill the appropriate abstraction level from a concrete software system, and relevant mathematical skills to develop, analyze, and validate the model. Indeed, the amount of knowledge required in both the problem domain and in the modeling techniques necessarily hinders their use in practice~\cite{woodside2007future}. 

Despite the promises that analytical performance modeling holds,  we are confronted with a high adoption barrier. A possible solution might be to derive the model \emph{automatically}. There has been much research into extending higher-level descriptions such as UML diagrams with performance annotations (using for example appropriate profiles such as MARTE~\cite{marte-spec}) from which both software artifacts and associated performance models are generated (see the surveys~\cite{DBLP:journals/tse/BalsamoMIS04,Koziolek2010634}). However, since systems are typically subjected to further modifications, the hard problem of keeping the model synchronized with the code arises~\cite{garcia2013obtaining}. This makes such model-driven approaches particularly difficult to use in general, especially in the context of fast-paced software processes characterized by continuous integration and development. 

\paragraph*{Main contribution} In this paper we propose a novel methodology where analytical performance models are automatically learned from a running system using execution traces. We focus on queueing networks (QNs), a formalism that has enjoyed considerable attention in the software performance engineering community, since it has been shown to be able to capture main performance-related phenomena in software architectures~\cite{tse:archopt}, annotated UML diagrams~\cite{DBLP:conf/wosp/BalsamoM05}, component-based systems~\cite{Koziolek2010634}, web services~\cite{1310688}, and adaptive systems~\cite{DBLP:conf/qosa/ArcelliCFL15,ase17}.

A QN is characterized by a number of parameters that define the following quantities: i) the behavior of each shared resource, such as its service demand and the concurrency level, which describe the amount of time that a client spends at the resource and the number of independent entities that can provide the service (e.g., number of threads in the pool or number of CPU cores), respectively; ii) the behavior of clients in terms of their operational profile, i.e., how they traverse the resources.      

Some of these parameters can be assumed to be known. For instance, the number of CPU cores is available from the hardware specification (or from the virtual-machine settings in a virtualized environment); the number of worker threads is a configuration parameter in most servers. Other parameters are more difficult to identify: the service demands, which depend on the execution behavior of the program that requests access to a shared resource; and the \emph{routing matrix}, which defines how clients (probabilistically) move between queuing stations. 

In our approach, the input is the set of shared resources and their concurrency level. The objective is to discover the QN model, i.e., the topology of the network and the service demands. 

Obviously, the problem of learning a mathematical model from data is not new. In the specific case of identifying parameters of a QN, a substantial amount of research gone into the problem of estimating service demands only (\cite{spinner2015evaluating}, see Section~\ref{sec:rel:qn} for a more detailed account of related work). Instead, we are not aware of approaches that deal with the estimation of both the service demands and the topology. This setting is a rather difficult one from a mathematical viewpoint because, as will be formalized later, routing probabilities and service demands appear as multiplicative factors in the dynamical equations that describe the evolution of a QN~\cite{bolch}. Since learning a QN can be understood as fitting the parameters to match these equations by some form of optimization, using both routing probabilities and service demands as decision variables will induce a \emph{nonlinear problem}, which is very difficult to handle in general. An additional problem to nonlinearity is that of scalability. This is due to the issue that the exact dynamical equations of a QN incur the well-known state explosion problem, because the number of discrete states to keep track of grows combinatorially with the number of clients and queuing stations.       

\paragraph{Learning method: recurrent neural networks.} To cope with both issues, we propose a learning method based on recurrent neural networks (RNNs) because of their ability to fitting nonlinear systems~\cite{machine_learning}. In particular, we develop a new architecture of the RNN which encodes the QN dynamics in an \emph{interpretable fashion}, i.e., by associating the weights of the RNN with QN parameters such as concurrency levels, routing probabilities, and service rates. A key instrument is the use of a compact system of \emph{approximate} (but still nonlinear) equations of the QN dynamics instead of the combinatorially large, but exact, original system of equations. Such approximation---called \emph{fluid} or \emph{mean-field}---consists in only one ordinary differential equation (ODE) for each station. It describes the time evolution of the \emph{queue length}, i.e., the number of clients contending for that resource. In practice, the fluid approximation provides an estimate of the average queue length of the underlying stochastic process. The QN approximation procedure is based on a fundamental result by Kurtz~\cite{kurtz-1970} and is well-known in the literature, e.g.,~\cite{Bortolussi2013317}. In the field of software performance, it has been used for the analysis of variability-intensive software systems~\cite{fase14,ase15} and for model-based runtime software adaption using online optimization~\cite{ase17} or satisfiability modulo theory approach~\cite{seams16}. This formulation has also been recently adopted for learning, but for service demands only~\cite{mascots18}, thus casting the problem into a (considerably) simpler quadratic programming one. 

The connection between RNN and ODEs is not new in the literature. In~\cite{pearlmutter1989learning}, authors have shown that recurrent neural networks can be thought of as a discretization of the continuous dynamical systems while, in~\cite{chen2018neural} a specialized training algorithm for ODEs has been recently proposed. However, despite the proliferation of works along this research direction, still there is no clear understanding of how to employ such artificial intelligence/machine learning techniques for supporting performance engineering tasks 
such as modeling, estimation, and optimization~\cite{icpe}.

\noindent\,\,\fbox{%
	\parbox{0.95\linewidth}{%
		The main technical contribution of this paper is to show that there is a direct association between the structure of the QN fluid approximation and standard activation functions and layers of an RNN. To the best of our knowledge, this is the first approach that formally unifies the expressiveness of analytical performance models with the learning capability of machine learning, contributing to positively answering the question whether ``\emph{AI will be at the core of performance engineering}''~\cite{icpe}.
	}%
}

The RNN is trained using time series of measured queue lengths at each service station. Its learned weights can be interpreted back as a QN with learned parameters, which can be used for predictive purposes. It is worth remarking that, in principle, one could learn a QN model by relying on a standard, \emph{black-box} RNN architecture by treating all the QN parameters (i.e., initial population, service demand, number of servers and routing probabilities) as input features of the learning algorithm. Unfortunately, this straightforward approach would require a considerable amount of input traces since the learning algorithm could not exploit the structural information about the problem. For instance, it would not be possible to do accurate what-if analyses by varying the value of a parameter if the network had not been trained with some input configurations where few variations of that parameter are considered. Moreover, in such setting, it would even be unclear which weights must be altered and how to reflect the changes into the model. 

Instead, here we report on the effectiveness and the generalizing power of our method by considering both synthetic benchmarks on randomly generated QNs, as well as a real web application deployed according to the load balancing architectural style. In both cases, we evaluate the degree of predictive power of the learned model in matching the transient as well as steady-state dynamics of unseen configurations (i.e., by varying the system workload, number of servers, and routing probabilities), reporting prediction errors less than 10\% across a validation set of 2000 instances.  

\paragraph*{Paper organization.} We provide some background about QNs in Section~\ref{sec:background}. The learning methodology is presented in Section~\ref{sec:learning}, which discusses how to encode a time-discretized version of the fluid approximation into an RNN where the weights represent the model parameters to identify. Section~\ref{sec:evaluation} presents the numerical evaluation on both the synthetic benchmarks and the real case study, providing implementation details on the RNN and on the used benchmark application. Section~\ref{sec:related} discusses further related work. Section~\ref{sec:conclusion} concludes.

\section{Background}\label{sec:background}

To make the paper self-contained, we present some background on QNs with the objective of motivating the fluid approximation as a deterministic estimator of average queue lengths, which will be used for the RNN encoding.
\subsection{Queuing Networks}
We assume \emph{closed} QNs, where clients keep circulating between queuing stations. A closed QN is formally defined by the following:
\begin{itemize}
    \item $N$: the number of clients in the network; 
    \item $M$: the number of queuing stations;
    \item $\srv = (s_1, \ldots, s_M)$: the vector of concurrency levels, where $s_i$ gives the number of independent servers at station $i$, with $1 \leq i \leq M$; 
    \item $\rate = (\mu_1, \ldots, \mu_M)$: the vector of service rates, i.e., $1/\mu_i>0$ is the mean service demand at station $i$, with $1 \leq i \leq M$;
    \item $\Pt = (\Pt_{i,j})_{1 \leq i,j \leq M}$: the routing probability matrix, where each element $\Pt_{i,j} \geq 0$ gives the probability that a client goes to station $j$ upon completion at station $i$;
    \item $\qlenFnc{0} = \left(\qlenFncS{0}{1}, \ldots, \qlenFncS{0}{M}\right)$: the \emph{initial condition}, i.e., $\qlenFncS{0}{i}$ is the number of clients at station $i$ at time 0. 
\end{itemize}
In a closed QN, the routing probability matrix is  \emph{stochastic} matrix, meaning that the sum across each row sums up to one.

\begin{figure}
\centering
\includegraphics[scale=0.40]{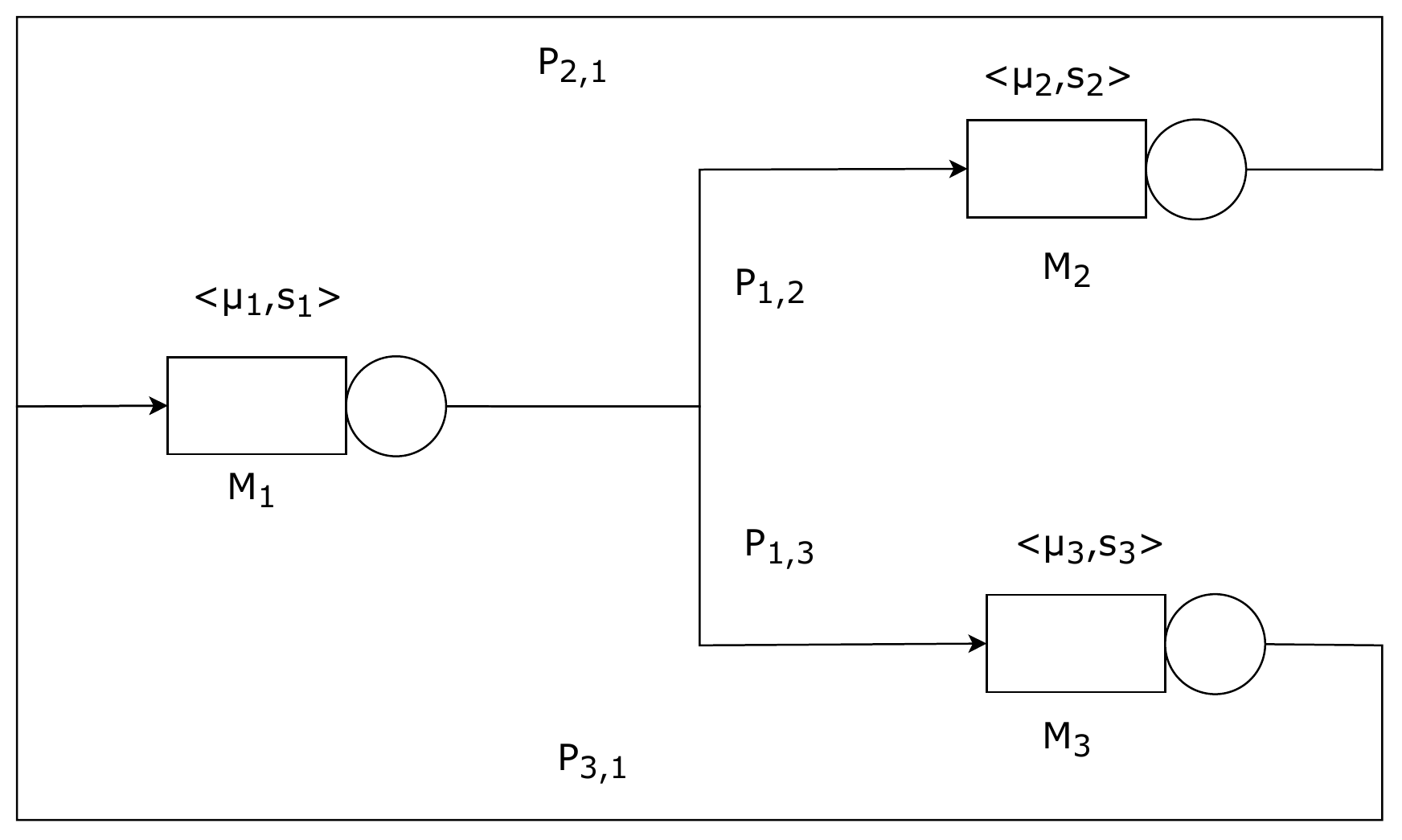}
\caption{Load balancing example}\label{fig:lb}
\end{figure}
\begin{example}
In the remainder of this section we use the QN in Fig.~\ref{fig:lb} as a running example. Depicted using the customary graphical representation, it represents a simple load-balancing system with $M=3$ stations. Requests from reference station \emph{\node{1}} are routed to two compute server stations \emph{\node{2}} and \emph{\node{3}} with probabilities $\Pt_{1,2}$ and $\Pt_{1,3}$, respectively. Upon service, a client returns back to station \emph{\node{1}}. An instantiation of this abstract model is discussed in Section~\ref{sec:evaluation}.  \qed
\end{example}   

\paragraph*{Markov chain semantics} The stochastic behavior of a QN is represented by a continuous-time Markov chain (CTMC) that tracks the probability of the QN having a given configuration of the queue lengths at each station. Informally, the CTMC is constructed as follows. A discrete CTMC state is a vector of queue lengths $\vec{X} = (X_1, \ldots, X_M)$. At each station $i$, if the number of clients $X_i$ is less than or equal to the number of servers $s_i$, then these proceed in parallel, each at rate $\mu_i$. Instead, if $X_i > s_i$ the number of clients that are queueing for service is $X_i - s_i$. When one client is serviced at station $i$, with probability $p_{ij}$ it goes to station $j$ to receive further service. This can be formalized by considering the well-known model of Markov population processes, whereby the CTMC transitions are described by jump vectors and associated transition functions from a generic state $\vec{X}$~\cite{Bortolussi2013317}. 

We define the jump vectors $h^{(ij)}$ to be the state updates due to clients moving to station $j$ upon service at $i$, and $q(\vec{X}, \vec{X} + h^{(ij)})$ the transition rate from state $\vec{X}$ to state $\vec{X} + h^{(ij)}$, where
$$\vec{X} + h^{(ij)} = (X_1, \ldots, X_i - 1, \ldots, X_j + 1, \ldots, X_M).$$
In other words, with the jump vector $h^{(ij)}$ we have that the number of clients at station $i$ is decreased by one, and, correspondingly, the number of clients at station $j$ is increased by one. 
Then, the CTMC is defined by:
\begin{equation}\label{eq:q}
q(\vec{X}, \vec{X} + h^{(ij)}) = \Pt_{i,j} \mu_i \min (X_i, s_i), \quad i,j = 1, \ldots, M.
\end{equation}

\begin{example}
In our running example, we have the jump vectors
\begin{align*}
h^{(12)} & = (-1, +1, 0) & h^{(13)} & = (-1, 0, +1) \\
h^{(21)} & = (+1, -1, 0) & h^{(31)} & = (+1, 0, -1) 
\end{align*} 
where the first row describes the updates due to a client being assigned to each compute server and the second row defines the client returning to the load balancer after service. For completeness we give the corresponding transitions:
\begin{align*}
q(\vec{X}, \vec{X} + h^{(12)}) & = \Pt_{1,2} \mu_1 \min (X_1, s_1) \\
q(\vec{X}, \vec{X} + h^{(13)}) & = \Pt_{1,3} \mu_1 \min (X_1, s_1) \\
q(\vec{X}, \vec{X} + h^{(21)}) & = \Pt_{2,1} \mu_2 \min (X_2, s_2) \\
q(\vec{X}, \vec{X} + h^{(31)}) & = \Pt_{3,1} \mu_3 \min (X_3, s_3) 
\end{align*}\qed
\end{example}

It is well known that a CTMC is completely characterized by the transitions (\ref{eq:q}) together with the initial condition $\qlenFnc{0}$. This formulation in terms of jump vectors allows for the efficient stochastic simulation of CTMCs~\cite{doi:10.1146/annurev.physchem}; indeed, we will use this technique to generate sample paths for the evaluation of our learning method on synthetic benchmarks in Section~\ref{sec:evaluation}. For our purposes, the main limitation of this CTMC representation is that the exact equations to analyze the probability distribution grow combinatorially with the number of clients and stations, as one needs to keep track of each possible discrete configuration of the queue lengths.

\subsection{Fluid Approximation}
The fluid approximation of a QN consists is an ODE system whose size is equal to the number of stations $M$, independently from the number of clients in the system. Informally, the ODE system can be built by considering the average impact that each transition has on the queue length at each station $k$. This is obtained by multiplying the $k$-th coordinate of each jump vector, $h^{(ij)}_k$, by the  function associated with the corresponding transition rate $q(\vec{X}, \vec{X} + h^{(ij)})$. Denoting by $\vec{x} = (\vec{x}_1, \ldots, \vec{x}_M)$ the variables of the fluid approximation, the ODE system is given by:
\begin{equation}\label{eq:odeInt}
\frac{d\qlenFncS{t}{k}}{dt} = \sum_{h^{(ij)}} h^{(ij)}_k q(\vec{x}(t), \vec{x}(t) + h^{(ij)}), \ k = 1, \ldots, M.
\end{equation}    

The solution for each coordinate, $\qlenFncS{t}{k}$, can be interpreted as an approximation of the average queue length at time $t$ as given by the CTMC semantics~\cite{Bortolussi2013317}. The theorems in~\cite{kurtz-1970} provide a result of asymptotic exactness of the fluid approximation, in the sense that the ODE solution and the expectation of the stochastic process become indistinguishable when the number of clients and servers is large enough.

Using~(\ref{eq:q}), the equations can be written as follows:
\begin{equation}\label{eq:ode}
\frac{d\qlenFncS{t}{k}}{dt} 
= \sum_{i≠k} \Pt_{i,k}\rate_i \min(\qlenFncS{t}{i}, \srv_i) +(\Pt_{k,k}-1)\rate_k \min(\qlenFncS{t}{k}, \srv_k)
\end{equation}  
where we have singled out the rates due to self loops $\Pt_{k,k}$.  
\begin{example}
The fluid approximation for the load balancer is:
\begin{align*}
\frac{d\qlenFncS{t}{1}}{dt} & = - \rate_1 \min(\qlenFncS{t}{1}, \srv_1) + \rate_2\min(\qlenFncS{t}{2}, \srv_2) \mathop{+} \\
& \qquad \ \mathop{+} \rate_3\min(\qlenFncS{t}{3}, \srv_3) \\  
\frac{d\qlenFncS{t}{2}}{dt} & = - \rate_2 \min(\qlenFncS{t}{2}, \srv_2) + \Pt_{1,2}\rate_1\min(\qlenFncS{t}{1}, \srv_1) \\
\frac{d\qlenFncS{t}{3}}{dt} & = - \rate_3 \min(\qlenFncS{t}{3}, \srv_3) + \Pt_{1,3}\rate_1\min(\qlenFncS{t}{1}, \srv_1)
\end{align*}\qed
\end{example}
Based on the solution to Eq.~\ref{eq:ode}, which directly provides queue-length estimates, one can derive other important performance metrics such as throughput, utilization, and response time. See, for instance~\cite{DBLP:journals/tse/TribastoneGH12,10.1109/TSE.2011.81} for a study of these results in a process algebra~\cite{pepa}, and~\cite{LNCS64160051,pepalqn,10.1109/TSE.2012.66} for  applications to layered queueing networks~\cite{10.1109/TSE.2008.74}.

In the remainder of this paper, we shall focus on QNs that do not have \emph{self loops} (i.e., a client served at a queue cannot re-enter the same queue immediately), i.e., $\Pt_{i,i} = 0$ for  $1 \leq i \leq M$. This is because we can show that, in the fluid approximation, for each $k$, $\Pt_{k,k}$ can be chosen freely as long as we adjust each $\Pt_{k,i}$ 
with $i \neq k$ and $\rate_k$. More formally, we can prove the following theorem.
\begin{thm}
\label{not_univoque}
For each $\pi \in [0,1)^{M}$, stochastic matrix $\Pt$ and $\rate\!≥\!0$ where \eqref{eq:ode} holds,
there exist $\hat{\Pt}$ and $\hat{\rate}$ such as for each $k$:
\begin{enumerate}[label=(\alph*)]
\item $\frac{d\qlenFncS{t}{k}}{dt} = \sum_{i≠k} \hat{\Pt}_{i,k}\hat{\rate}_i$ $\min(\qlenFncS{t}{i}, \srv_i) $\\$\hphantom{spspspspspspspsp} +(\hat{\Pt}_{k,k}-1)\hat{\rate}_k \min(\qlenFncS{t}{k}, \srv_k)$; 
\item $\hat{\Pt}_{k,k} = \pi_k$; 
\item $\sum_i\hat{\Pt}_{k,i}=1$; 
\item $\forall i \ \hat{\Pt}_{k,i}≥0$; 
\item $\hat{\rate}_k ≥ 0$.
\end{enumerate}
\end{thm}
\begin{proof}
Available in Appendix~\ref{app:not_univoque_proof}.
\end{proof}
Thus, using the fluid approximation, for each network with self loops there is another one without them which cannot be distinguished. To identify a specific network among them, we need to know the self-loop values.

\section{Learning methodology}\label{sec:learning}

As apparent in both Equations (\ref{eq:q}) and (\ref{eq:ode}), a QN features routing probabilities and service demands as multiplicative factors in the defining dynamical equations. If we wish to learn a QN by assuming that both quantities are unknown, we are faced with a nonlinear (i.e., polynomial) optimization problem. Here we propose an RNN in order to estimate these parameters. 
We develop an RNN architecture which encodes the QN dynamics in an \emph{interpretable fashion}, i.e., by associating the weights of the RNN with QN parameters such as concurrency levels, routing probabilities, and service rates. 

\subsection{ODE Discretization}
We first obtain a time-discrete representation of the fluid approximation such that each time step is associated with a layer of the RNN. In matrix notation, for an arbitrary QN the fluid approximation  is given by:
\[
	\frac{d \qlenFnc{t}}{dt} = - \rate \min\left(\qlenFnc{t}, \srv \right)  + \Pt^T \rate \min(\qlenFnc{t}, \srv) 
\]
where $\qlenFnc{t}$ is the $M$-dimensional vector of queue lengths at time $t$. 
We consider a finite-step approximation of the above ODE for a small $\Delta t$, obtaining: 
\[
\qlenFnc{t+\Delta t}  = \qlenFnc{t} + \Delta t \cdot \big(- \rate \min\left(\qlenFnc{t}, \srv \right) + \rate\Pt \min(\qlenFnc{t}, \srv)\big)
\]
Finally, this  can be rewritten as
\begin{equation}\label{eq:finite}
\qlenFnc{t+\Delta t}  = \qlenFnc{t} + \Delta t \cdot \useFnc{t} \cdot \left(\rate \odot \left(\Pt - \id\right)\right)
\end{equation}
where $\useFnc{t} = \min\left(\qlenFnc{t}, \srv \right)$, $\id$ is the identity matrix of appropriate dimension, and $\odot$ is the operator where if $\mathbf{C} = \mathbf{a} \odot \mathbf{B}$, then $\mathbf{C}_{i,j} = \mathbf{a}_{i} \cdot \mathbf{B}_{i,j}$. 

\begin{figure}[t]
	\centering
	\includegraphics[width=0.5\linewidth]{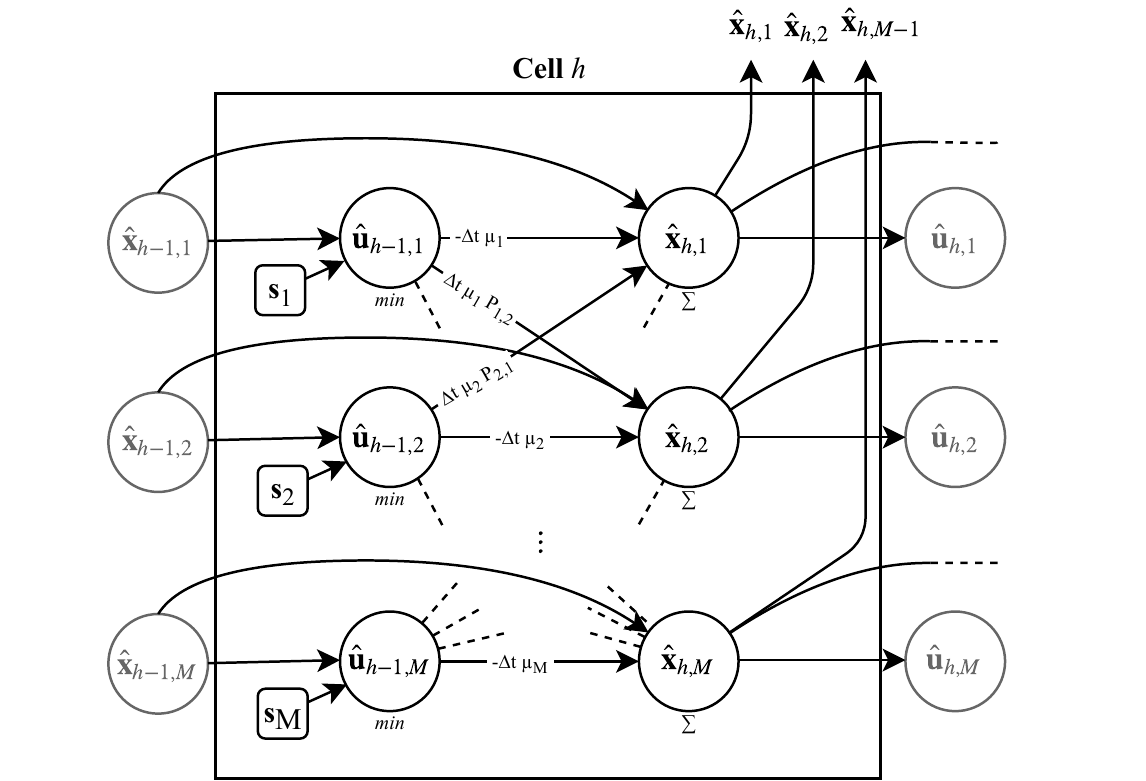}
	\caption{RNN encoding}
	\label{fig:encoding}
\end{figure}
\subsection{RNN Encoding}
The discretization (\ref{eq:finite}) of the fluid approximation of the QN admits a direct encoding as an RNN. It consists of an $M$-dimensional input layer $\qlens{0}$ that corresponds to the initial condition of the QN. The RNN has $H-1$ cells, with the $h$-th cell computing the estimate of the queue length at time $h \Delta t$, denoted by $\qlens{h}$ (see Fig.~\ref{fig:encoding}). That is, the $h$-th cell computes the quantity $\qlens{h} = \qlens{h-1} + \Delta t \cdot \use{h-1} \cdot \left(\rate \odot \left(\Pt - \id\right)\right)$, where, according to (\ref{eq:finite}), $\use{h-1}$ estimates $\useFnc{(h-1)\Delta t}$ as 
$\use{h-1} = \min\left(\srv, \qlens{h-1}\right)$. 

\begin{figure}[!ht]
	\centering
	\begin{subfigure}{0.45\linewidth}
		\centering
		\includegraphics[width=\linewidth]{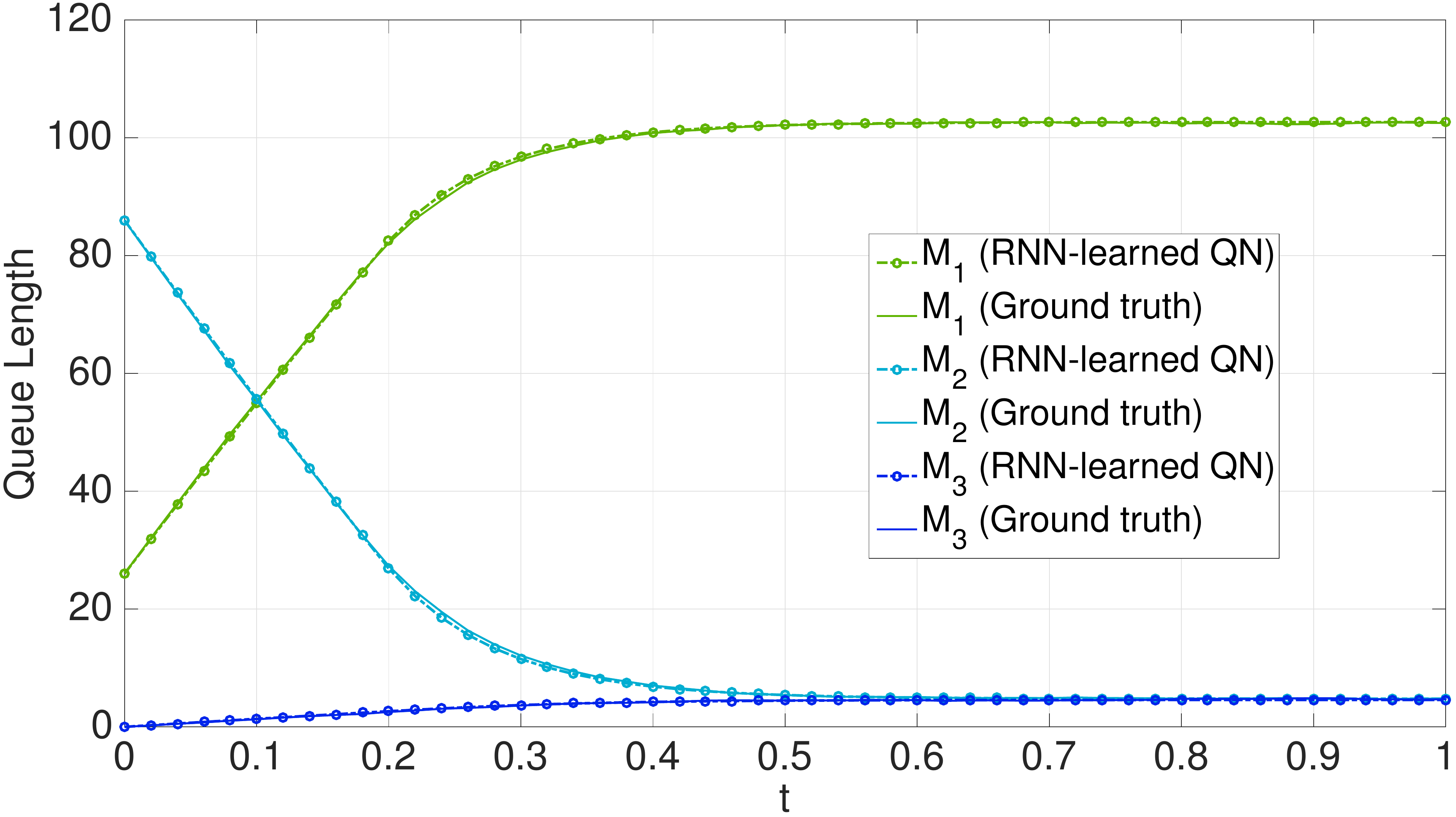}
		\caption{}
		\label{fig:learning_illustrative}
	\end{subfigure}
	\begin{subfigure}{0.45\linewidth}
		\centering
		\includegraphics[width=\linewidth]{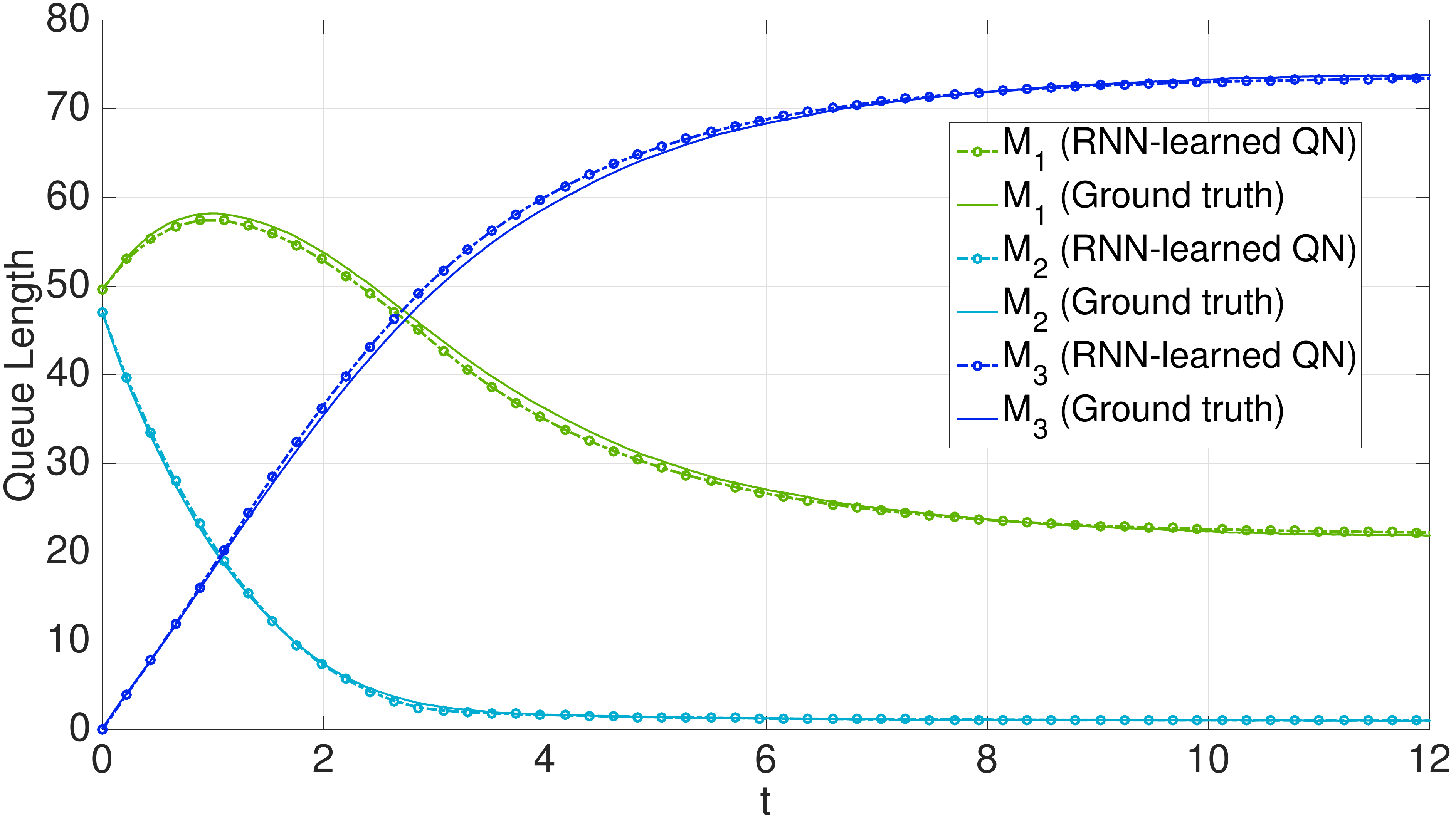}
		\caption{}
		\label{fig:whatif_illustrative}
	\end{subfigure}
	\caption{Numerical evaluation of the running example (see Figure~\ref{fig:lb}). Comparison between simulations of the queue lengths using the RNN-learned QN (marked lines) and the ground-truth one (straight lines) in two different cases: a) a trace used for training ($\mathit{err} = 0.69\%$) with initial population $x(0) = (26, 86, 0)$ and concurrency levels ($s_1$=$1000$, $s_2$=$30$, $s_3$=$25$) b) what-if analysis under unseen initial population vector and unseen concurrency levels ($s_1$=$1000$, $s_2$=$6$, $s_3$=$1$) and initial population $x(0) = (49, 47, 0)$, causing a significant change in the dynamics ($\mathit{err} = 1.49\%$).}
	\label{fig:illustrativeExample-result}
\end{figure}


With this set up, we will have to learn the matrix $\Pt$ (made of $M (M-1)$ weights, since the diagonal is empty) and the vector $\rate$ (made of $M$ weights).

The main goal of this methodology is to learn the actual parameters of the network. Therefore, we enforce some feasibility constraints, namely we require that $\Pt$ rows sum up to $1$ (such that $\Pt$ is a stochastic matrix), absence of self loops and $\rate ≥ 0$ (such that the speed of the stations is non-negative). The non-negativity of the weights is enforced in the framework by clamping the candidate values within the range $[0,\infty)$; stochasticity of $\Pt$ is guaranteed by dividing each weight by the sum of the weights in the corresponding row; the absence of self loops is achieved by setting $\forall i, \Pt_{i,i} = 0$ as a constant. This approach puts our work in the \emph{explainable machine learning} research area \cite{explainable_ai}, and it allows us to link each learned parameter with its role in the system. This link allows us to predict the behavior of the system under new conditions (\emph{what-if} analysis). In contrast, a traditional approach to neural networks would not impose a model and constraints on the parameters, hence giving a read-only model which cannot be clearly interpreted. Indeed, without a direct association between parameters and physical quantities, we cannot study the system under new conditions unless learning a new model.

\begin{example}
The RNN encoding for the $h$-th cell (i.e., the queue length transient evolution at time $h \Delta t$) of our running example is:
\begin{align*}
\use{h-1,1} & = \max\left(\srv_1, \qlens{h-1,1} \right)\\
\use{h-1,2} & = \max\left(\srv_2, \qlens{h-1,2} \right)\\
\use{h-1,3} & = \max\left(\srv_3, \qlens{h-1,3} \right)\\
\qlens{h,1} = &\;\qlens{h-1,1} + \Delta t \left(-\rate_1 \use{h-1,1} +\rate_2 \Pt_{2,1} \use{h-1,2} +\rate_3 \Pt_{3,1} \use{h-1,3} \right)\\
\qlens{h,2} = &\;\qlens{h-1,2} + \Delta t \left(\rate_1 \Pt_{1,2} \use{h-1,1} -\rate_2 \use{h-1,2} +\rate_3 \Pt_{3,2} \use{h-1,3} \right)\\
\qlens{h,3} = &\;\qlens{h-1,3} + \Delta t \left(\rate_1 \Pt_{1,3} \use{h-1,1} +\rate_2 \Pt_{2,3} \use{h-1,2} -\rate_3 \use{h-1,3} \right)
\end{align*}\qed
\end{example}

\subsection{Input data} The RNN is trained over a set of traces. Each trace is made of $H$ vectors, indicated as $\qlen{0}, \qlen{1}, ..., \qlen{H-1} \in \mathbb{R}^{M}_{\geq 0}$. The $i$-th component $\qlen{h,i}$ of each vector $\qlen{h}$ represents a sample of the queue length of each station $i$ at time $h \cdot \Delta t$. Since, as discussed, the fluid approximation can be interpreted as an estimator of the average queue lengths, each trace used in the learning process consists of measurements averaged over a number of independent executions started with the same initial condition; different traces rely on{\tiny } different initial conditions to exercise distinct behaviors of the system.

\subsection{Learning function} 

The learning error function, denoted by $\mathit{err}$, aims to minimize the difference between the queue lengths estimated by the RNN, $\qlens{h}$, and the measurements $\qlen{h}$. It is defined as follows:
\begin{equation}\label{eq:err}
\mathit{err} = \frac{\max_{h = 1}^{H-1} {\|\qlen{h} - \qlens{h}\|}}{2 N} \cdot 100
\end{equation}
where $\| \cdot \|$ indicates the L1 norm. Essentially, it is a maximum relative error. Indeed, since we are studying closed QNs with fixed $N$ circulating clients, the quantity $\|\qlen{h} - \qlens{h}\|/(2 N)$ intuitively measures the proportion of clients (relative to their total number $N$) that are ``misplaced'' (i.e., which are allocated in a different station) at each time step. Since a misplaced client is counted twice (once when missing in a queue and once when is extra in another queue), we divide the norm by 2. Then, the overall error $\mathit{err}$ computes the maximum of such misplacements across all times.

%

\begin{figure}[t]
	\centering
	\begin{subfigure}{0.6\linewidth}
		\centering
		\includegraphics[width=\linewidth]{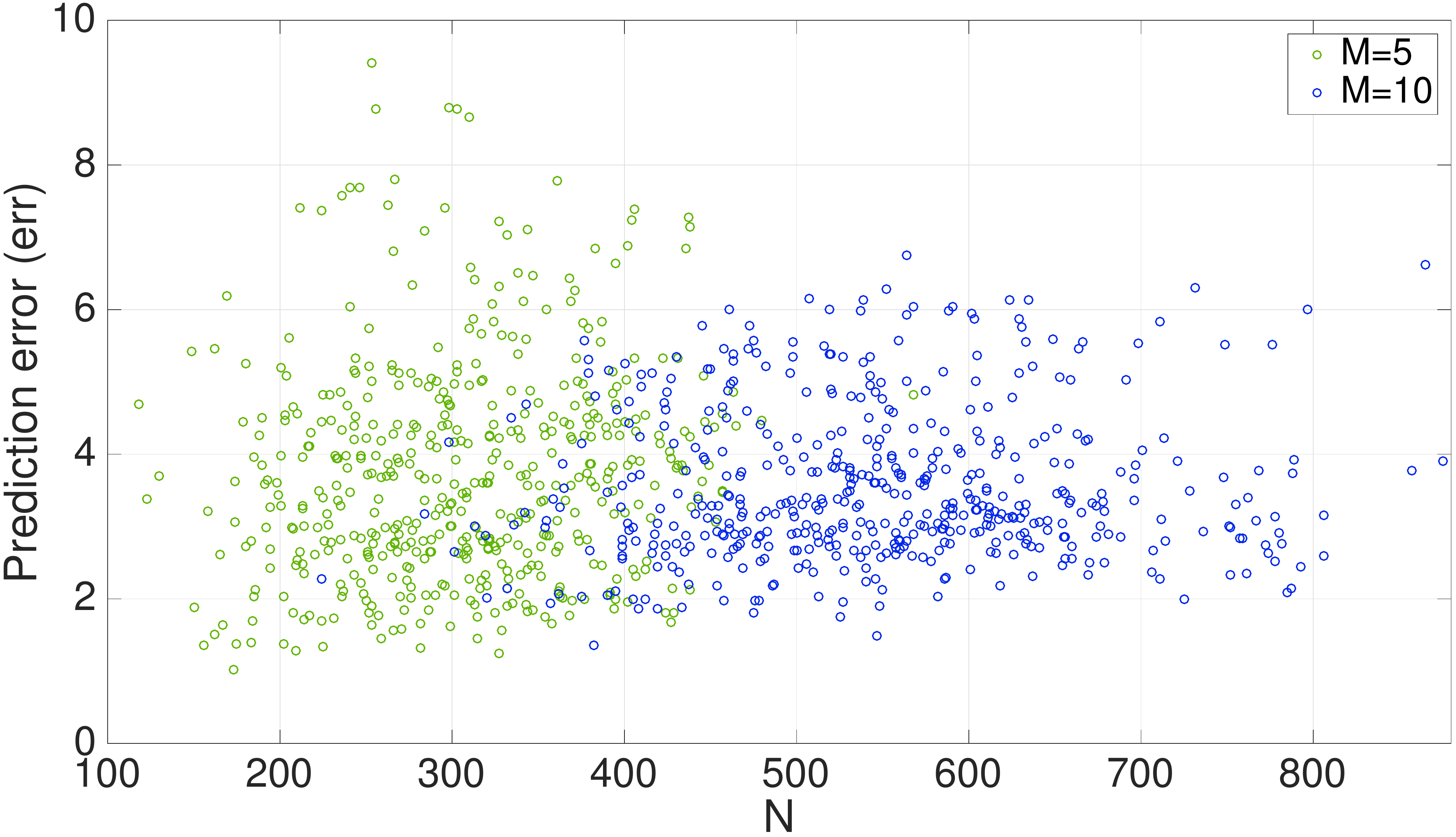}
		\caption{}
		\label{fig:scattera}
	\end{subfigure}
	\begin{subfigure}{0.24\linewidth}
		\centering
		\includegraphics[width=\linewidth]{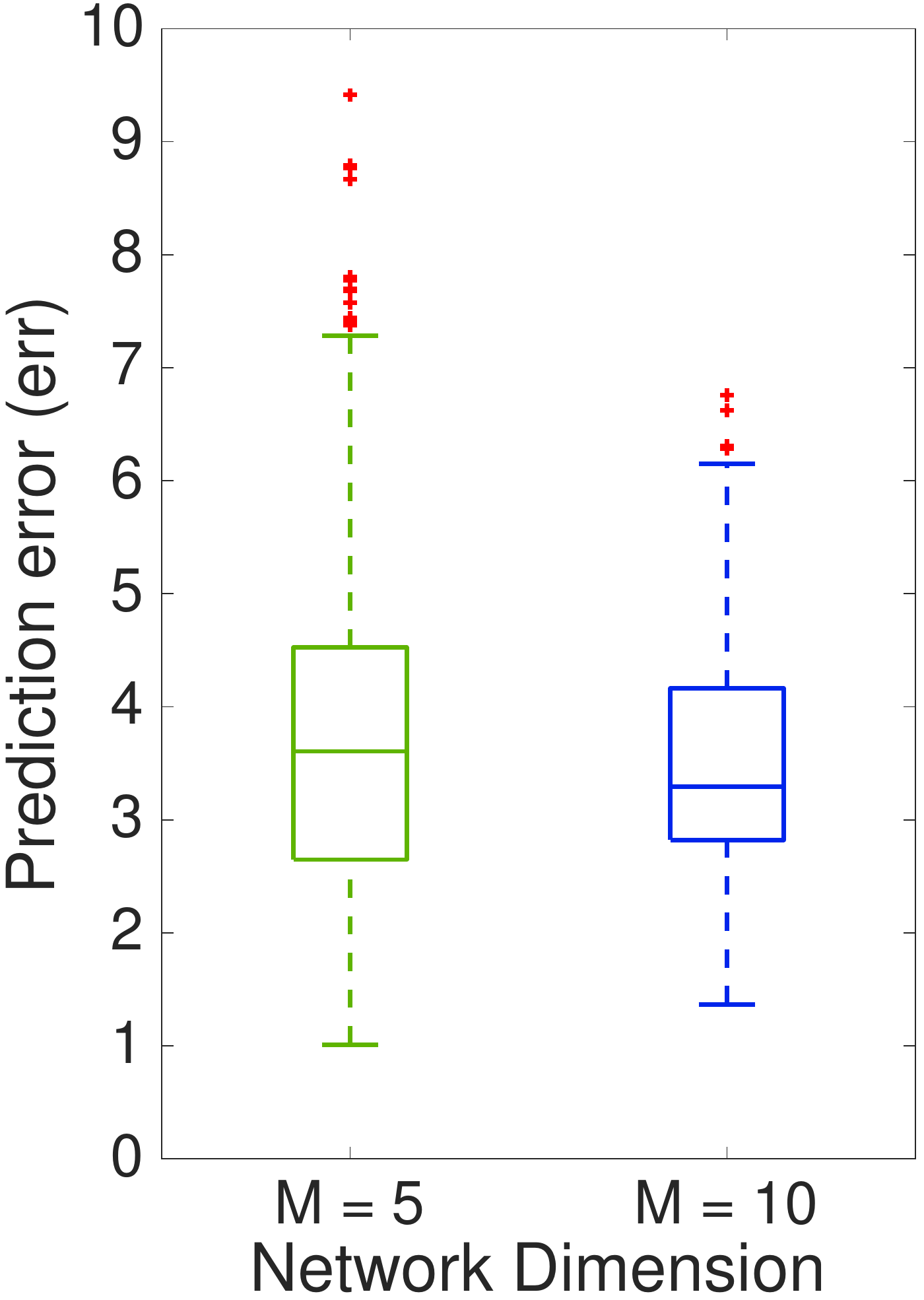}
		\caption{}
		\label{fig:scatterb}
	\end{subfigure}
	\caption{a) Prediction error of the what-if instances where each randomly generated QN is tested with 100 unseen initial population vectors, distinguished in colors with respect to the network size $M$. The x-axis $N$ is the total number of clients in the network which is scatter-plotted against the prediction error defined in Eq.~(\ref{eq:err}). b) Statistics on the prediction error. In each box-plot, the line inside the box represents the median error, the upper and lower side of the box represent the 25th and 75th percentiles, while the upper and lower limit of the dashed line represent the extreme points not to be considered outliers, and in red we depict the outliers (12 with M=5, 4 with M=10).}\label{fig:scatter}
\end{figure}

\begin{figure*}
	\centering
	\begin{subfigure}{0.195\linewidth}
		\centering
		\includegraphics[width=\linewidth]{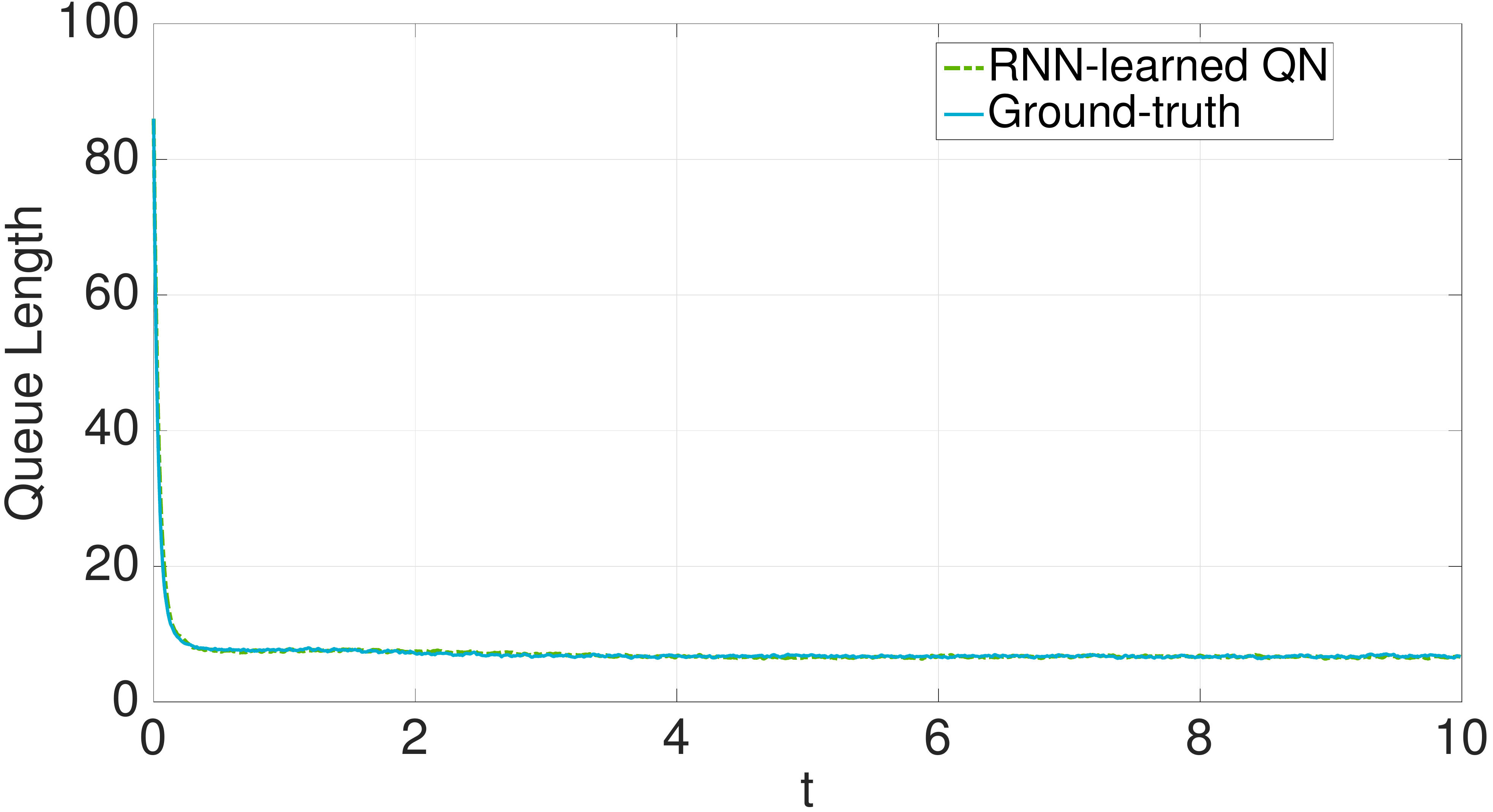}
		\caption{Station 1}
	\end{subfigure}
	\begin{subfigure}{0.195\linewidth}
		\centering
		\includegraphics[width=\linewidth]{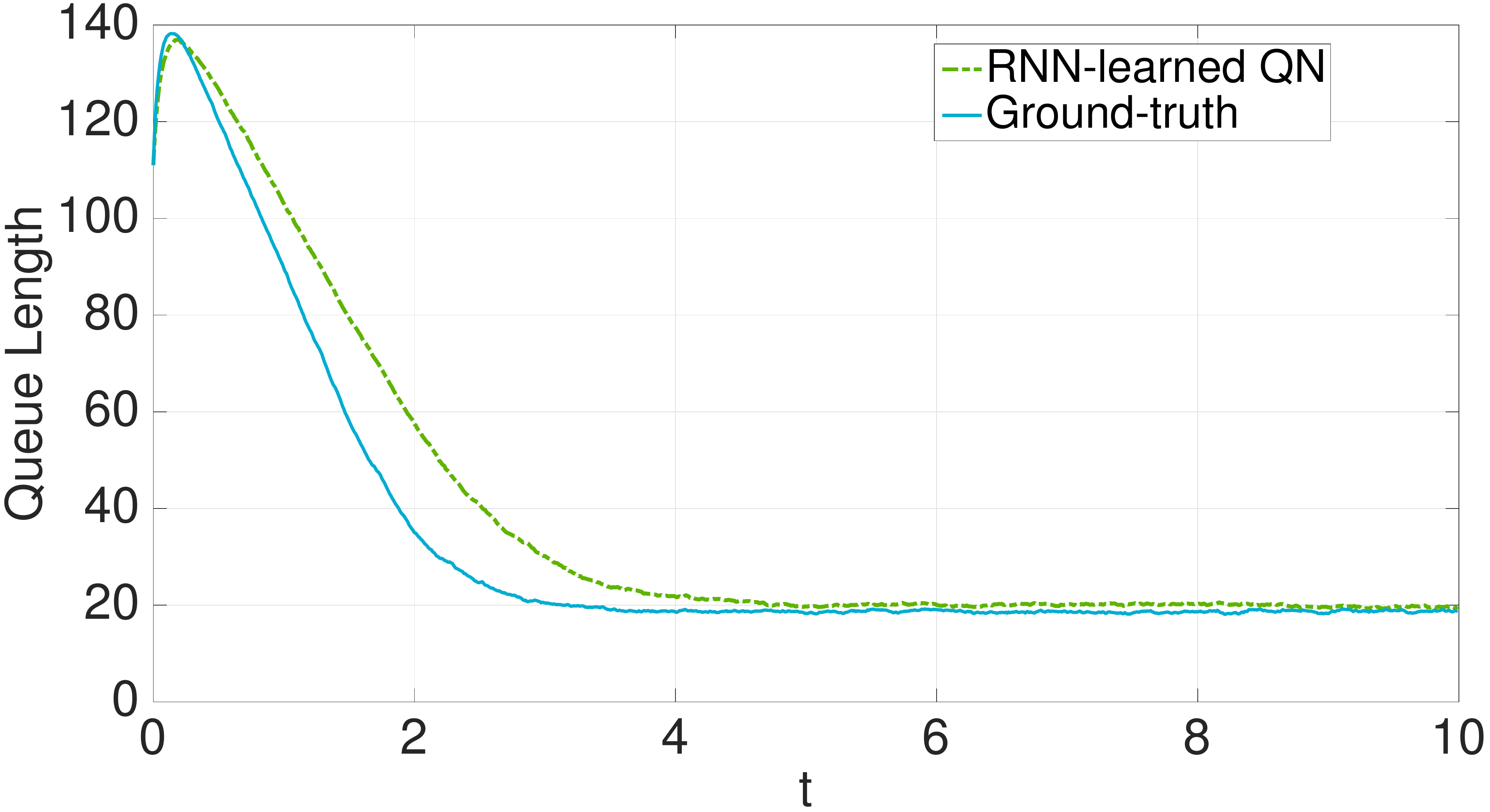}
		\caption{Station 2}
	\end{subfigure}
	\begin{subfigure}{0.195\linewidth}
		\centering
		\includegraphics[width=\linewidth]{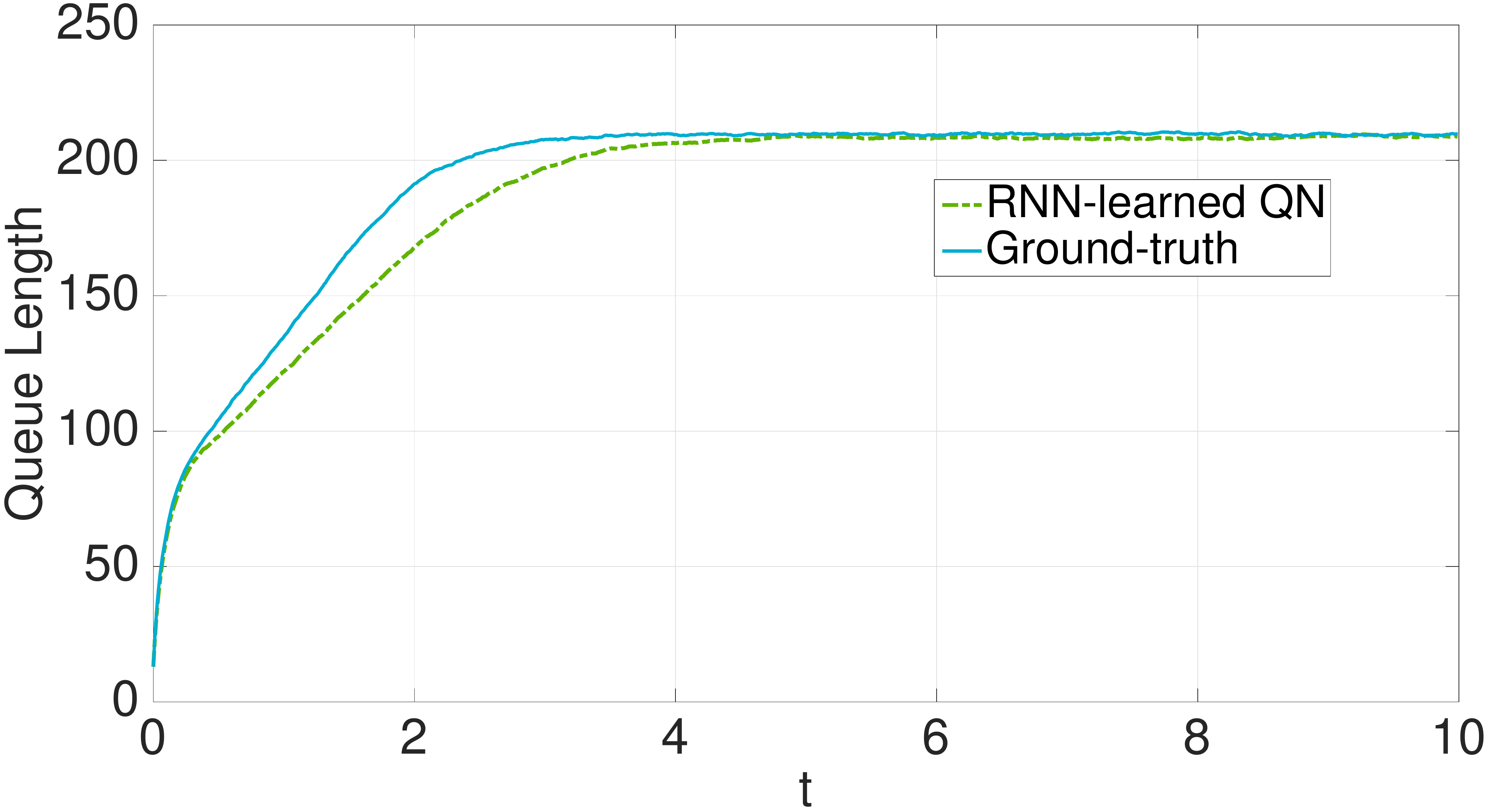}
		\caption{Station 3}
	\end{subfigure}
	\begin{subfigure}{0.195\linewidth}
		\centering
		\includegraphics[width=\linewidth]{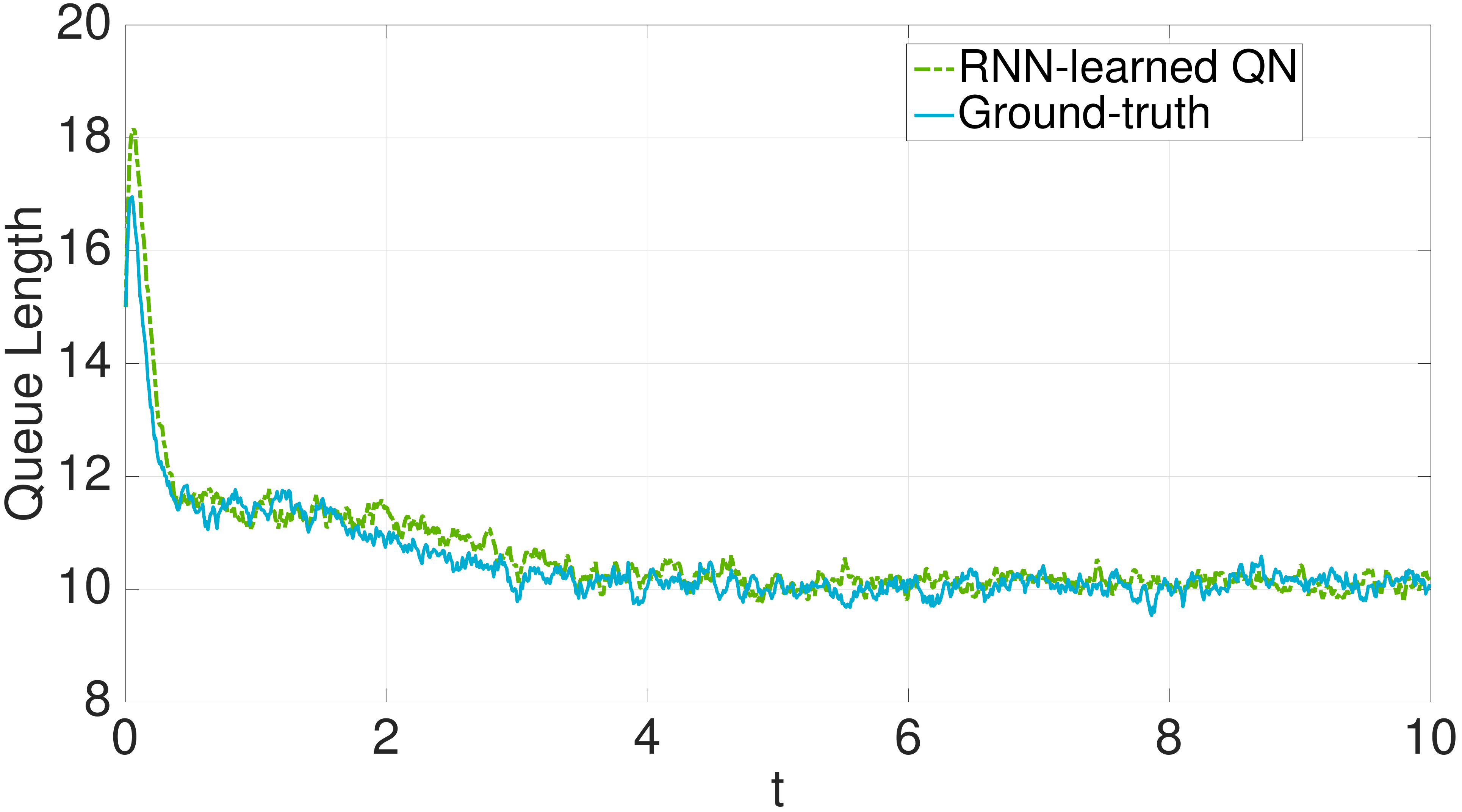}
		\caption{Station 4}
	\end{subfigure}
	\begin{subfigure}{0.195\linewidth}
		\centering
		\includegraphics[width=\linewidth]{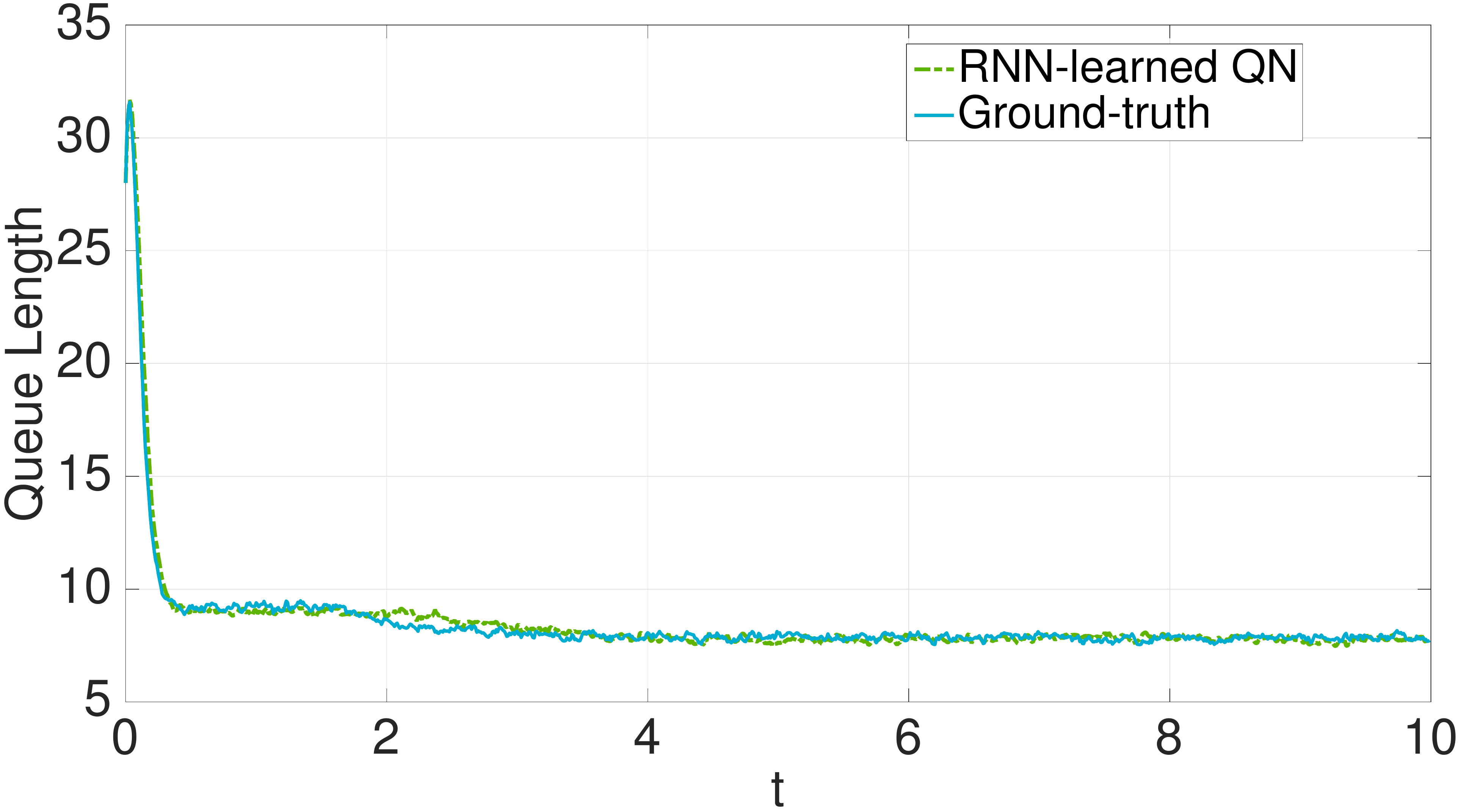}
		\caption{Station 5}
	\end{subfigure}
	\caption{Comparison between the ground-truth queue lengths and those predicted by the RNN-learned QN on the test case that induced the maximum prediction error among the what-if over population (error: 9.41\%). The error was attained on a randomly generated QN with $M = 5$ stations, using the unseen initial population vector (86,111,13,15,28). The straight line represents the ground-truth dynamics of the QN model; the dashed line represents the evolution of the RNN-learned QN.}
	\label{fig:qlen}
\end{figure*}

\begin{example}
Let us consider our running example by fixing ground-truth parameters as follows. During the learning phase, we studied the system with $\srv = (1000, 30, 25)$ and predicted the behavior with $\srv = (1000, 6, 1)$, while we kept $\Pt$ and $\rate$ unchanged at
	\[
		\Pt = \left[\begin{array}{ccc}0 & 0.5 & 0.5\\1 & 0 & 0\\1 & 0 & 0\end{array}\right]\quad\quad\quad\quad\rate = (1, 11, 11)
	\]
Using the experimental set up that will be discussed in the next section, we generated the training dataset by collecting 50 traces, one for a different randomly generated initial population vector. Each trace was the average of  500 independent simulations  recording the transient evolutions of the queue lengths.  Figure~\ref{fig:illustrativeExample-result} reports the comparison between queue lengths of the RNN-learned QN and the ground-truth one, showing very good accuracy on an instance of the training set (Figure~\ref{fig:learning_illustrative}) as well as high predictive power of the model under unseen initial populations and concurrency levels which cause bottleneck shift and considerable longer transient dynamics (Figure~\ref{fig:whatif_illustrative}).

\end{example}


\begin{figure}[t]
	\centering
	\begin{subfigure}{0.6\linewidth}
		\centering
		\includegraphics[width=\linewidth]{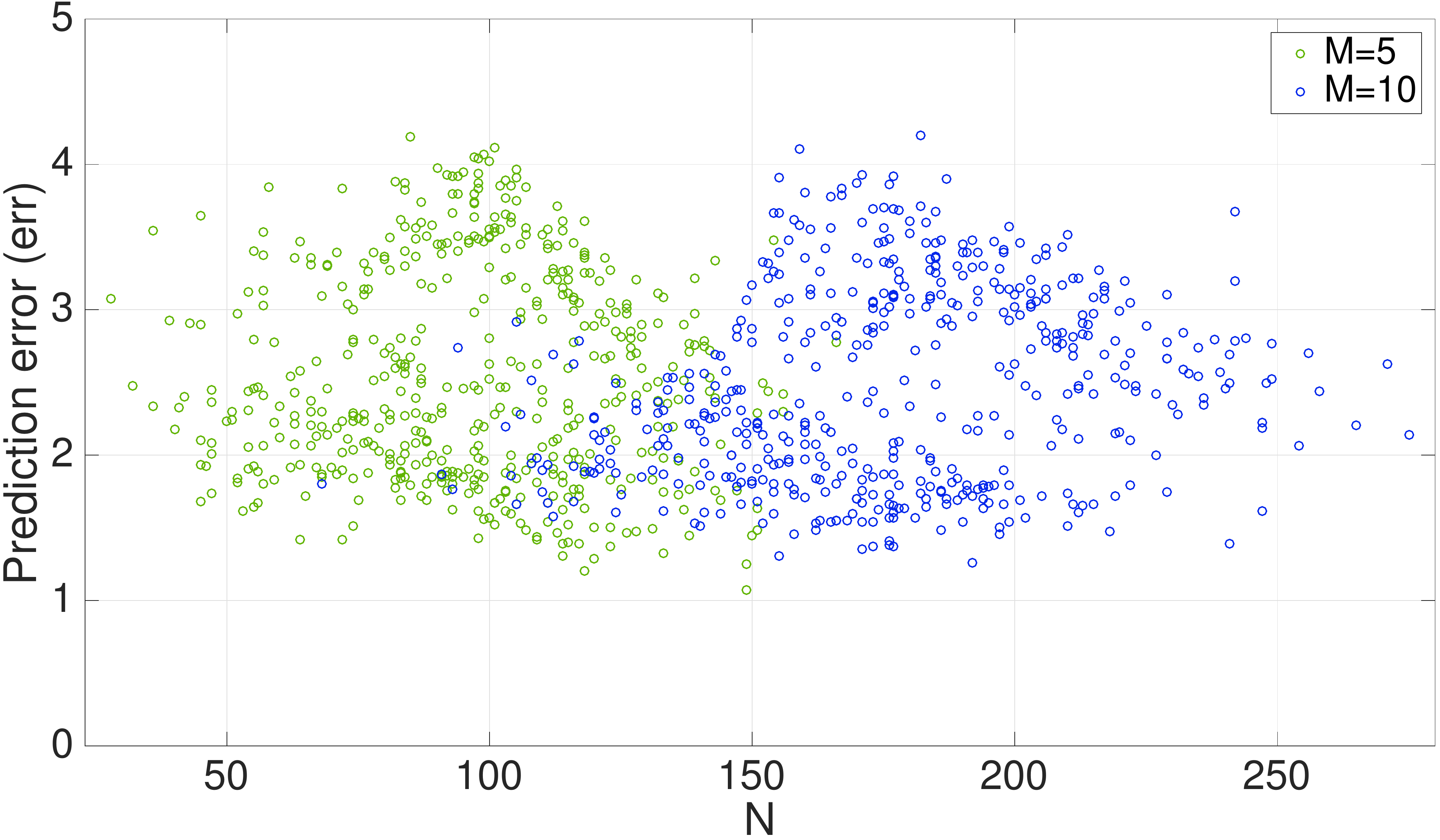}
		\caption{}
		\label{fig:scatter2a}
	\end{subfigure}
	\begin{subfigure}{0.24\linewidth}
		\centering
		\includegraphics[width=\linewidth]{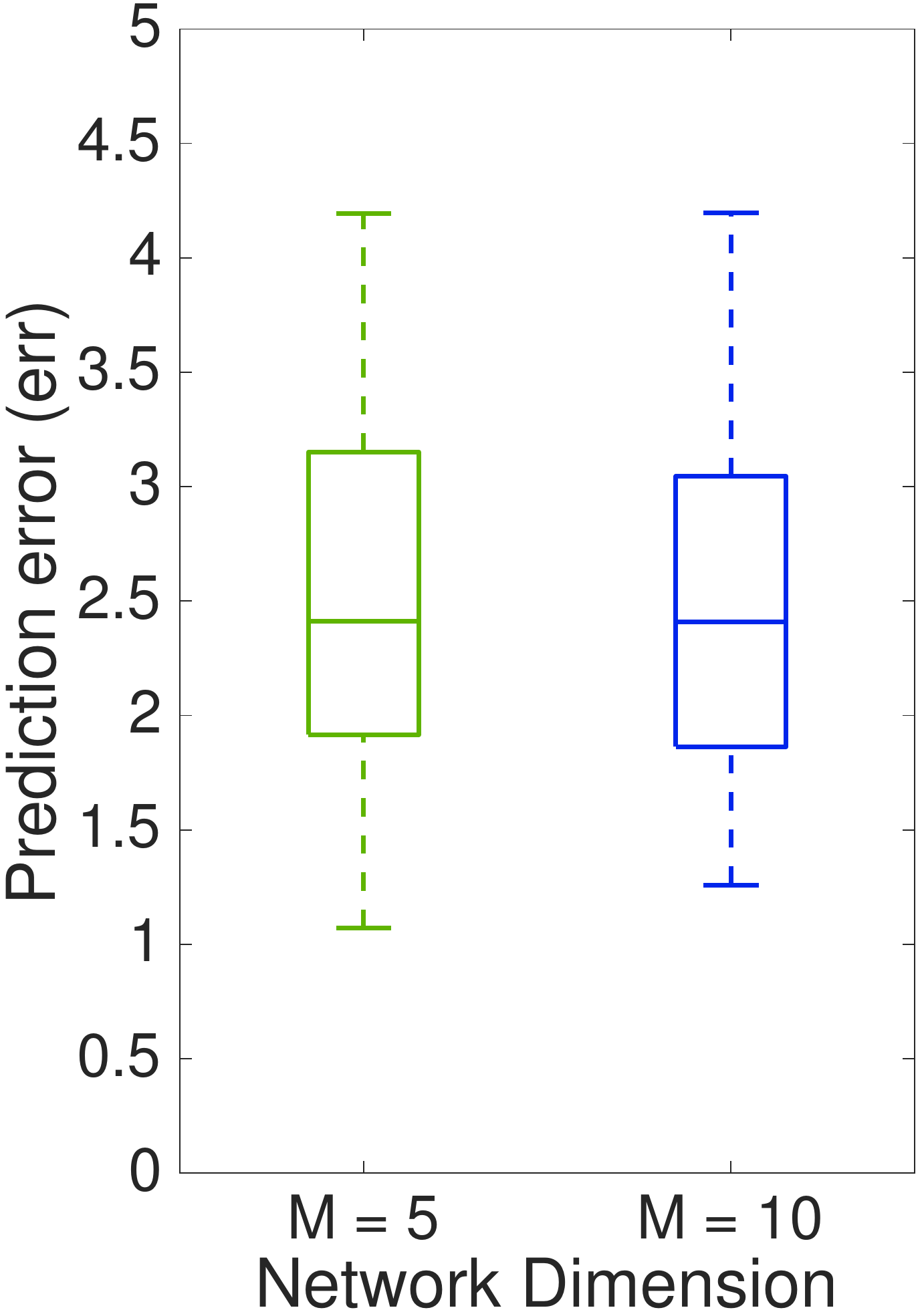}
		\caption{}
		\label{fig:scatter2b}
	\end{subfigure}
	\caption{a) Prediction error of the what-if instances by changing the concurrency level of the most utilized station in each of the randomly generated QNs. b) Statistics on the prediction error.}\label{fig:scatter2}
\end{figure}

\section{Numerical Evaluation}\label{sec:evaluation}

In this section we evaluate the effectiveness of the proposed approach by considering both synthetic benchmarks and a real case study. For all our tests, the RNNs were implemented using the Keras framework~\cite{chollet2015keras} with the TensorFlow backend~\cite{tensorflow2015}. Learning was performed using a machine running the 4.15.0-55-generic~Linux kernel on a Intel(R)~Xeon(R)~CPU E7-4830 v4 machine at 2.00GHz with 500~GB of RAM.

\subsection{Synthetic case studies}

\paragraph{Set-up} For our synthetic tests we considered randomly generated networks of size $M = 5$ and $M = 10$. For each case, we generated 5 QNs by uniformly sampling at random the entries of the routing probability matrices, the service rates in the interval $[4.0, 30.0]$, and the concurrency levels in the interval $\{15, 16, \ldots, 30\}$.
For the training of each QN, we generated 100 traces, each being the average over 500 independent stochastic simulations (generated using Gillespie's algorithm \cite{doi:10.1146/annurev.physchem}). Each trace exercised the model with a distinct initial population vector such that the number of clients at each station was drawn uniformly at random from $\{ 0, \ldots, 40\}$; as a result, the total number of clients in the network varies across traces. For each network, learning was performed by equally splitting the 100 traces for training and validation, iterating Adam~\cite{Adam} with learning rate equal to $0.05$, until the error computed on the validation set did not improve by at least $0.01\%$ in the last 50 iterations. On average, the learning took 74 minutes and 86 minutes for the cases $M=5$ and $M=10$, respectively.
 
\paragraph{Discretization methodology} Two important parameters are the length of the trace, i.e., the time horizon $T$ of the stochastic simulations, and the choice of the discretization interval $\Delta t$; these are related with the number of cells in the RNN $H$ by $T = (H-1) \Delta t$. Longer time horizons lead to larger simulation (hence, training) runtimes. Too short traces might not expose the full dynamics of the system. Further, following basic facts about ODE discretization~\cite{petzold:ode:book}, the interval $\Delta t$ should be chosen small enough such that no important dynamics is lost across two successive time steps; thus, longer time horizons might need more time steps, hence more cells in the RNN. It is worth remarking that these considerations are model-specific. That is, the choice of such hyper-parameters must be carefully done depending on the specific QN under study.  

For the synthetic case studies, we set $T=10$ and $\Delta t = 0.01$, hence $H=1000$.
\paragraph{Predictive power} We evaluate the predictive power of the learned QNs by performing two distinct ``what-if'' analyses under unseen configurations, by changing populations of clients and the concurrency levels of the stations, respectively.

\begin{figure*}
	\centering
	\begin{subfigure}{0.195\linewidth}
		\centering
		\includegraphics[width=\linewidth]{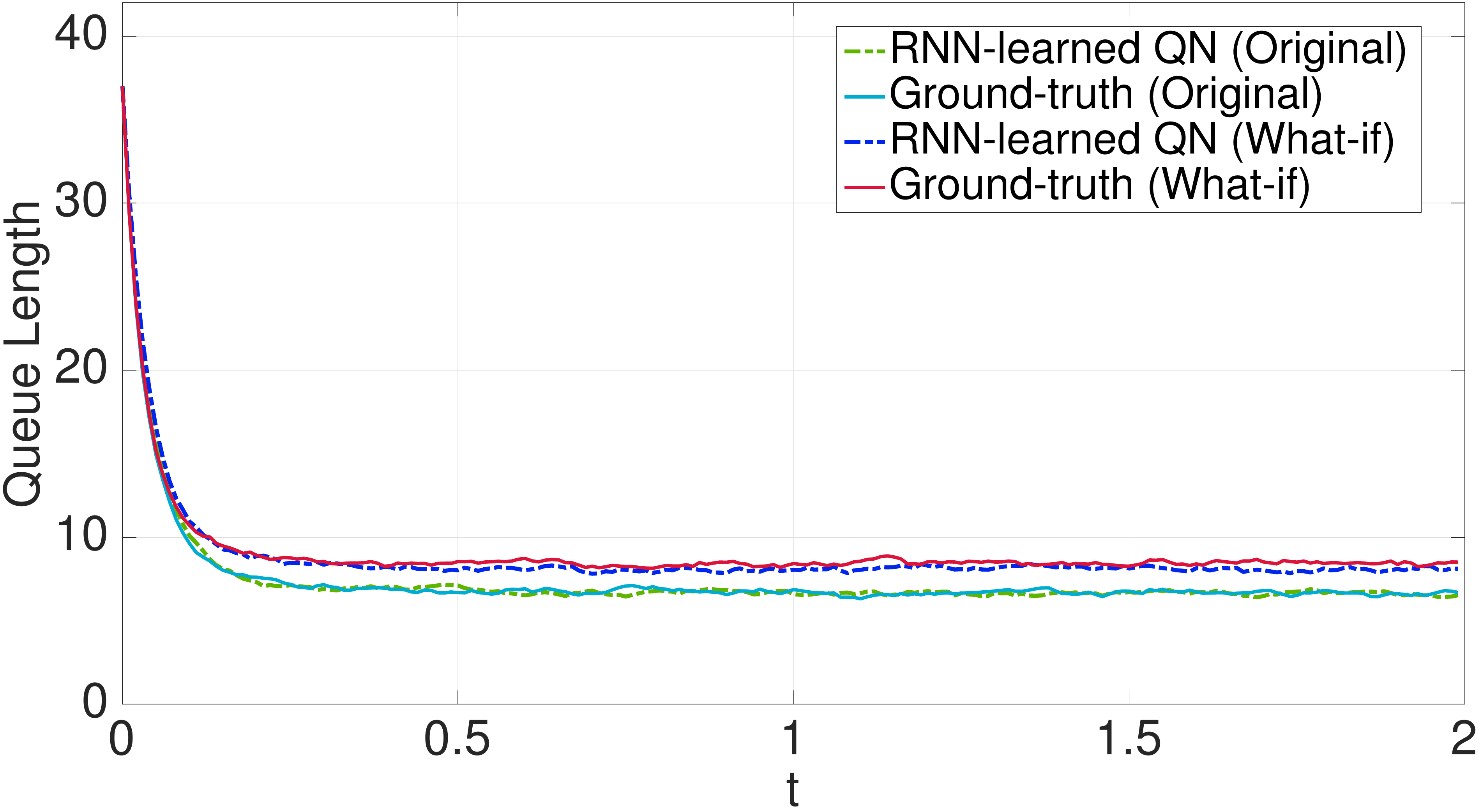}
		\caption{Station 1}
	\end{subfigure}
	\begin{subfigure}{0.195\linewidth}
		\centering
		\includegraphics[width=\linewidth]{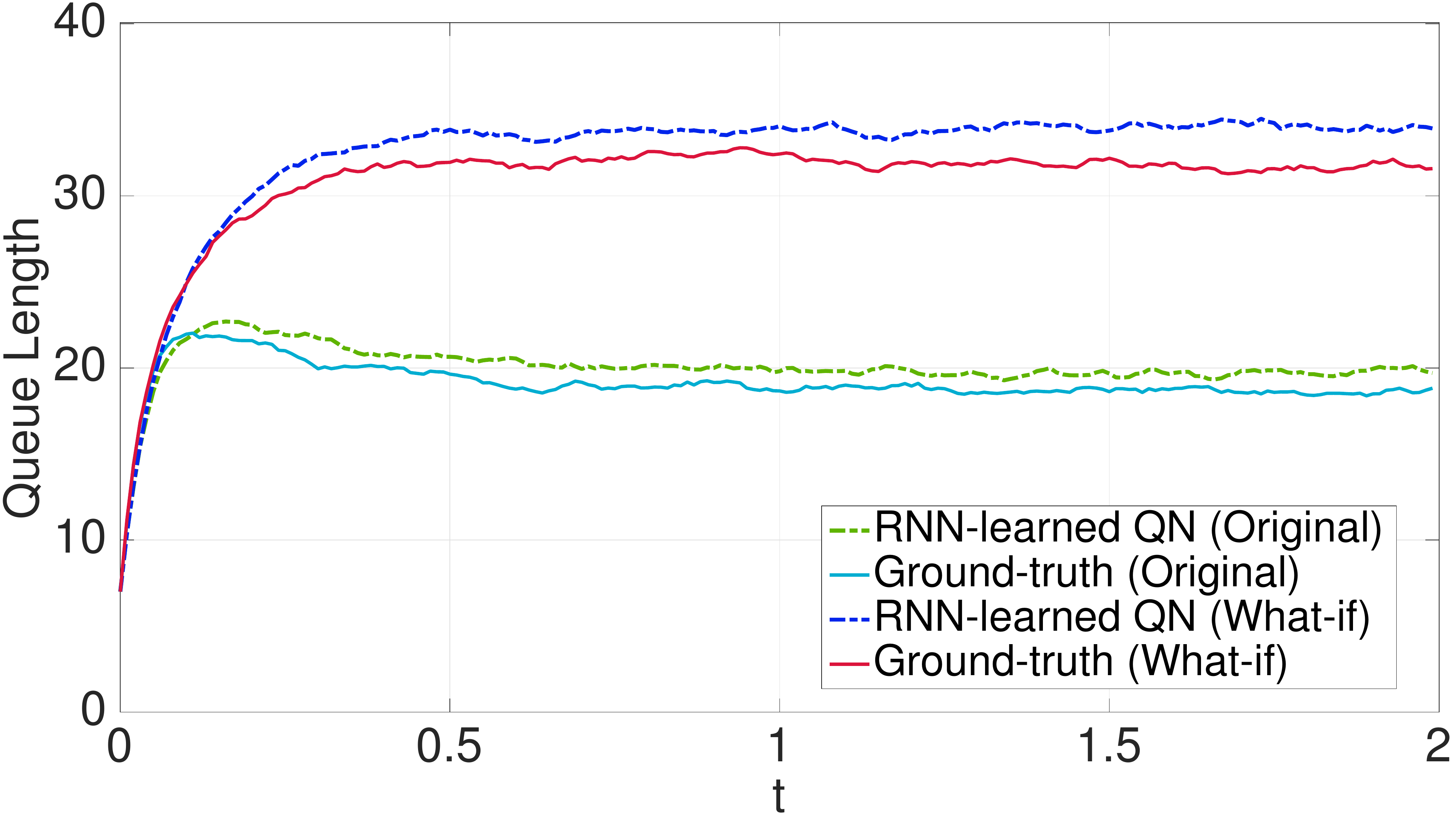}
		\caption{Station 2}
	\end{subfigure}
	\begin{subfigure}{0.195\linewidth}
		\centering
		\includegraphics[width=\linewidth]{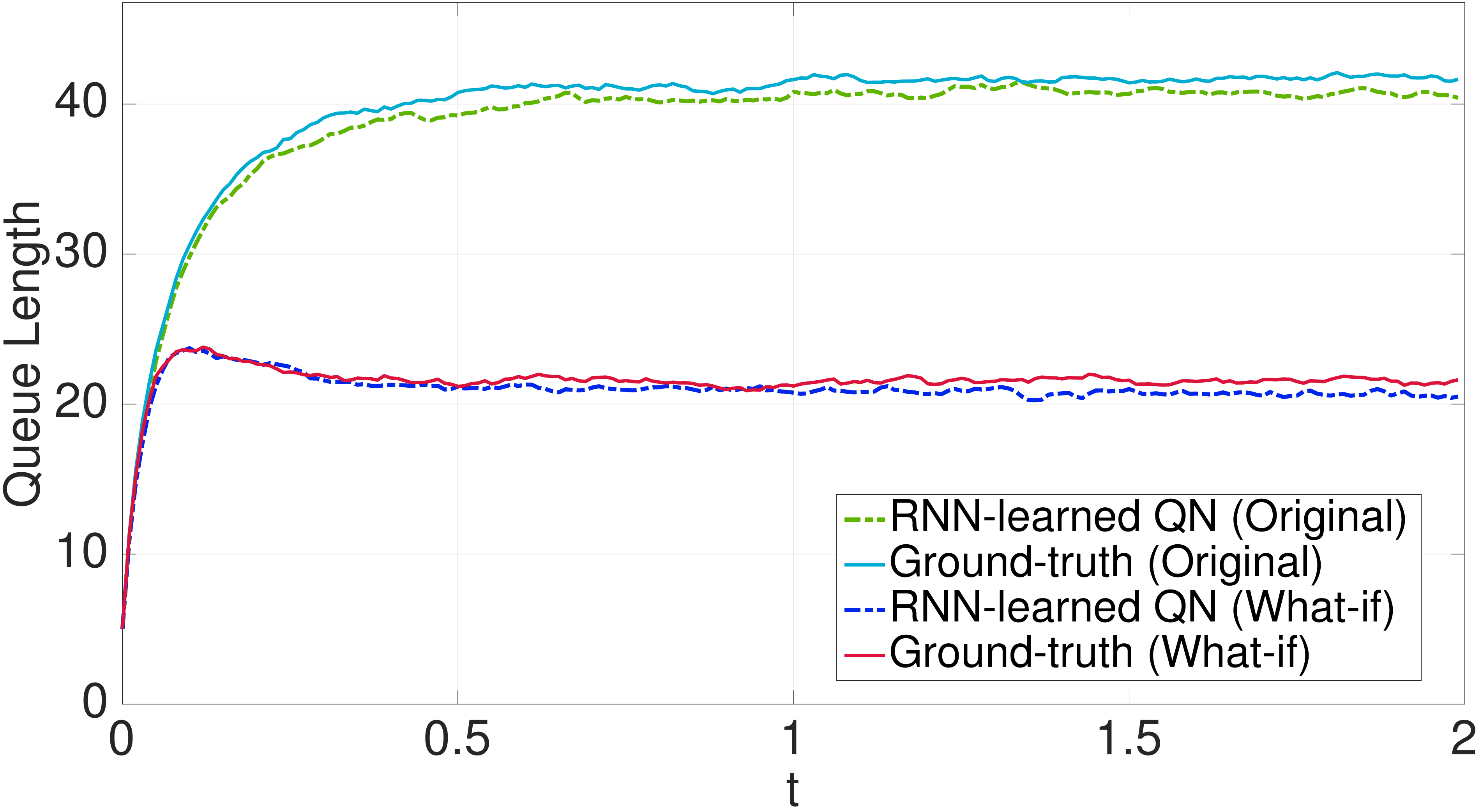}
		\caption{Station 3}
	\end{subfigure}
	\begin{subfigure}{0.195\linewidth}
		\centering
		\includegraphics[width=\linewidth]{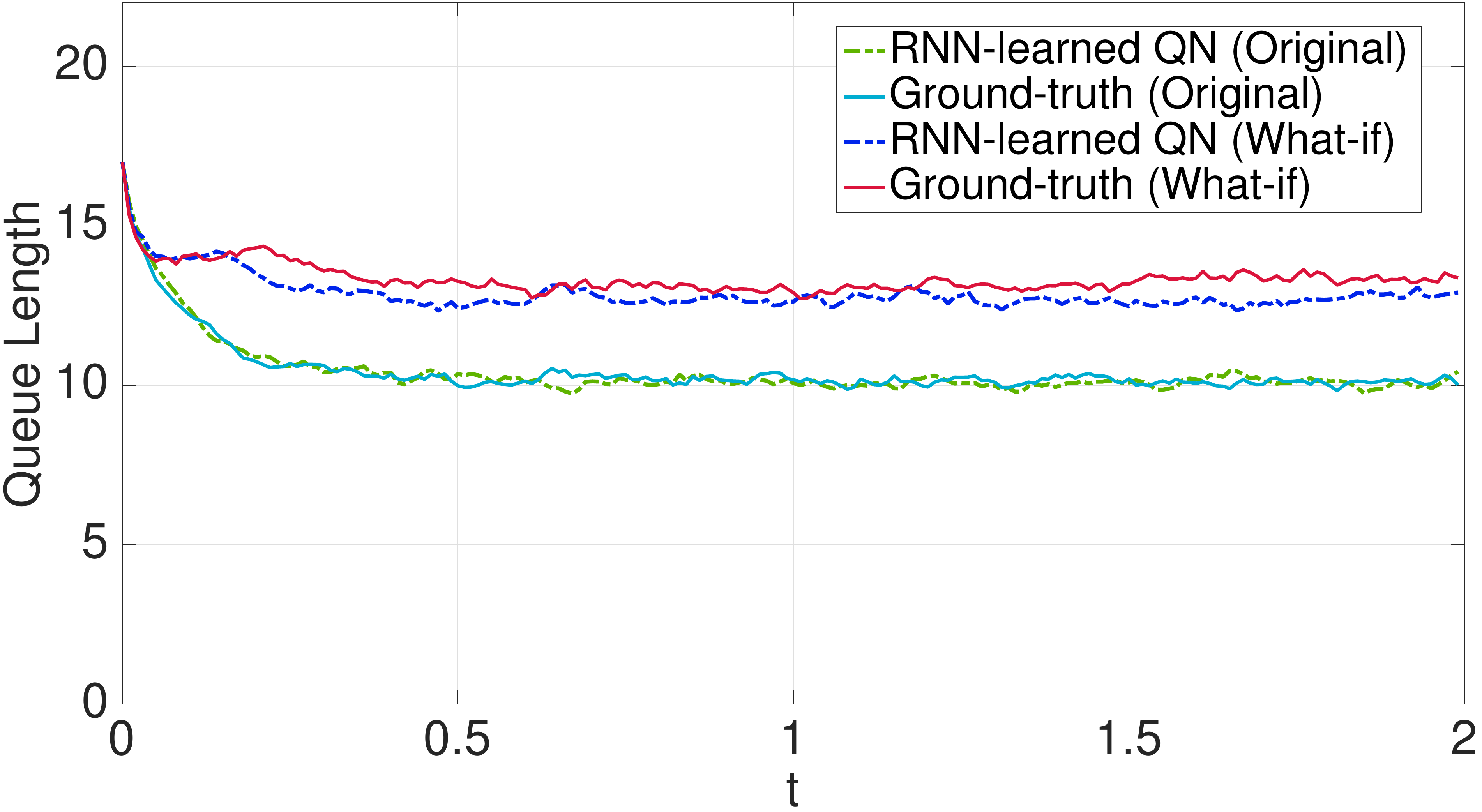}
		\caption{Station 4}
	\end{subfigure}
	\begin{subfigure}{0.195\linewidth}
		\centering
		\includegraphics[width=\linewidth]{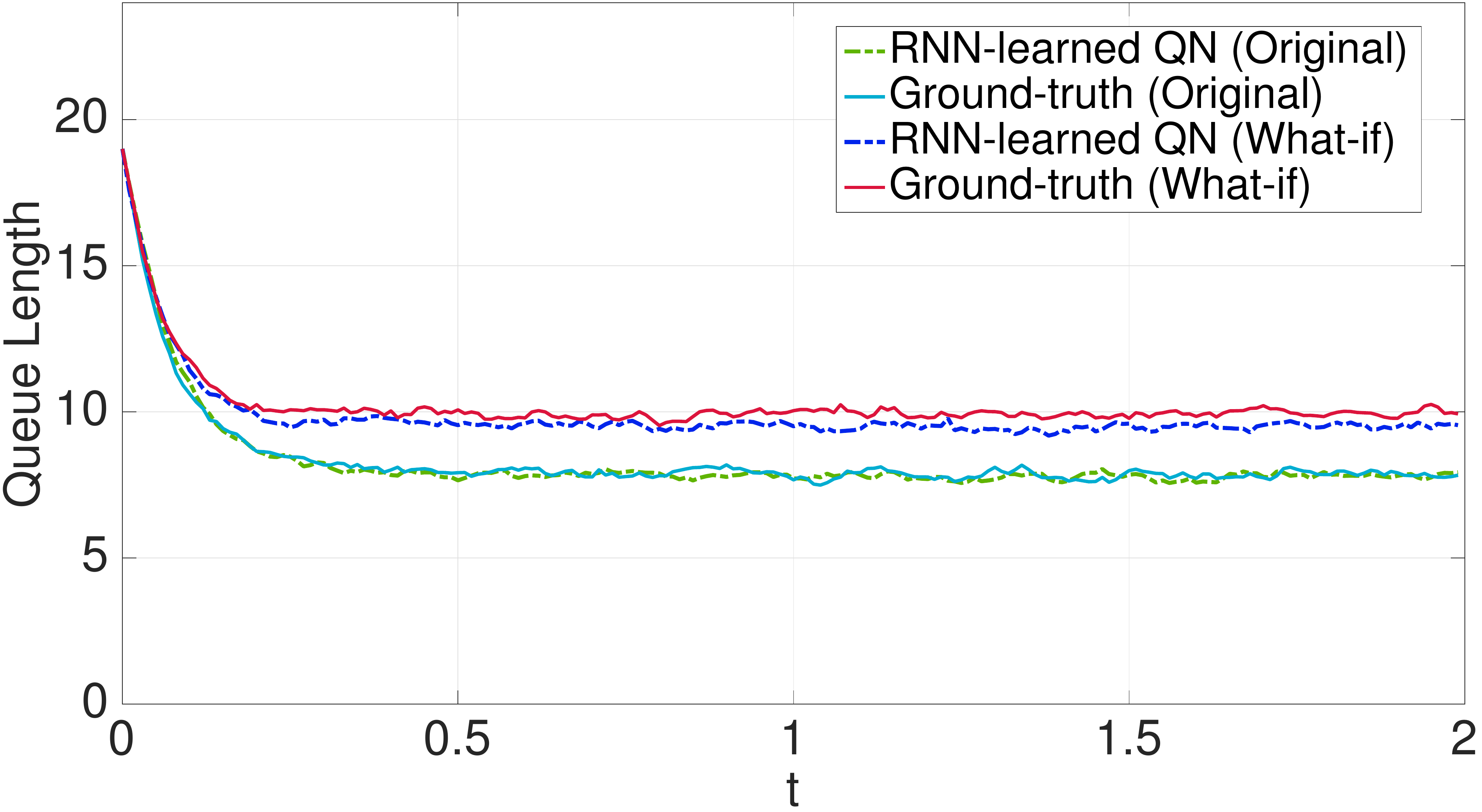}
		\caption{Station 5}
	\end{subfigure}
	\caption{Comparison between the ground-truth queue lengths and those predicted by the RNN-learned QN on the test case that induced the maximum prediction error (4.19\%), before and after the what-if change of server concurrency. The error was attained on a randomly generated QN with M = 5 stations. The cyan line denotes the averages under the original conditions (before the what-if change) with the ground-truth QN; the green line gives the predictions of the RNN-learned QN with the original values; the red line shows ground-truth simulations with the unseen number of servers for the bottleneck station (increased from 17 to 37); the blue line shows the averages after the what-if change for the RNN-learned QN.}
	\label{fig:whatifServer}
\end{figure*}

\paragraph{What-if analysis over client population} We tested each of the  randomly generated QNs with 100 new initial population vectors that were not used in the learning phase. We compared the averages (over 500 stochastic simulations) of the the ground-truth queue-length dynamics with those produced by the RNN-learned QN with those unseen initial conditions. 

Figure~\ref{fig:scattera} shows a scatter plot of the prediction error with respect to the total number of clients circulating in the system,  reporting errors less than $10\%$ in all cases. The box-plots in Figure~\ref{fig:scatterb} show that there is no statistically significant difference between the errors for the diffrent sized models. Figure~\ref{fig:qlen} compares the predicted and ground-truth queue lengths for the instance with the maximum prediction error, showing a very good generalizing power for the queue-length dynamics at all stations. 

\paragraph{What-if over concurrency levels.} To validate the predictive power under varying concurrency levels, for each generated QN we found the station with the highest ratio between the steady state queue length and its number of servers (\emph{bottleneck}), and added servers in steps of 20 to this station until it was not the bottleneck anymore. Then we compared the dynamics of the ground-truth model (i.e., simulated with the original $\Pt$ and $\rate$ but with the new server concurrency levels) against those obtained by simulating the learned model with the new server concurrency levels. We considered the notion of prediction error as shown in Equation~\eqref{eq:err}.

Figure~\ref{fig:scatter2a} shows the results of this what-if study, reporting a prediction error less than $5\%$ across all instances. Also in this case, there is no statistically significant difference in the error statistics depending on the network size $M$ (see Figure~\ref{fig:scatter2b}). Figure~\ref{fig:whatifServer} plots the comparison of the queue-length dynamics of the what-if instance (i.e., with an unseen server concurrency level) that reported the maximum prediction error (i.e., $4.5\%$) against the original ones (i.e., prediction error of $3.1\%$). We can appreciate that the unseen concurrency levels do change the QN behavior dramatically, effectively switching the bottleneck from station 3 to station 2.    

This result does support the combination of machine learning and white-box performance models by showing that, once learned, the QN can be used for evaluating the behavior of the model under execution scenarios for which the QN has not been trained. 


\begin{figure*}
	\centering
	\includegraphics[width=0.75\linewidth]{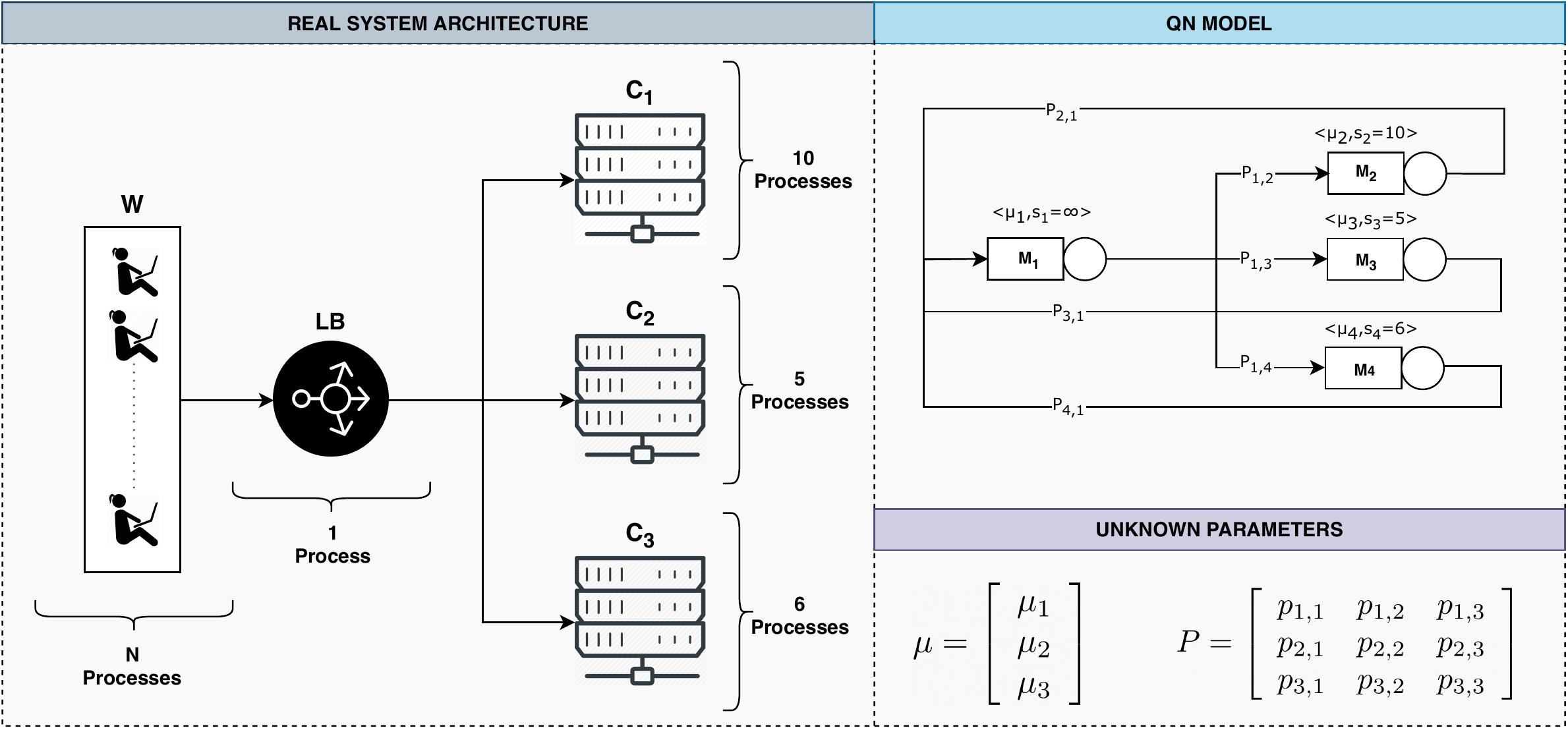}
	\caption{Case study architecture.}
	\label{fig:casestudy}
\end{figure*}

\subsection{Real case study}\label{sec:realCaseStudy}

\paragraph{Set-up} The benchmark used in this evaluation is based on an in-house developed web application that serves user requests with an input dependent load. We deployed the target application as a NodeJs~\cite{tilkov2010node} load-balancing system with three replicas.
Figure~\ref{fig:casestudy} (left) depicts the system architecture. Component \node{W} represents the reference station, where clients enter the system by issuing requests to the load balancer \node{LB}, which redistributes them across the web servers uniformly. In the real system, such uniform assignment is achieved by fixing equal \emph{weights} to the target nodes. Components \node{C1}, \node{C2}, and \node{C3} represent the three web-server instances devoted to the actual processing of user requests (e.g., producing an HTML page). Each node in the Figure~\ref{fig:casestudy} is annotated with its concurrency level (i.e., the number of available processes), which we considered fixed parameters. 


\begin{figure*}
	\centering
	\begin{subfigure}{0.246\linewidth}
		\centering
		\includegraphics[width=\linewidth]{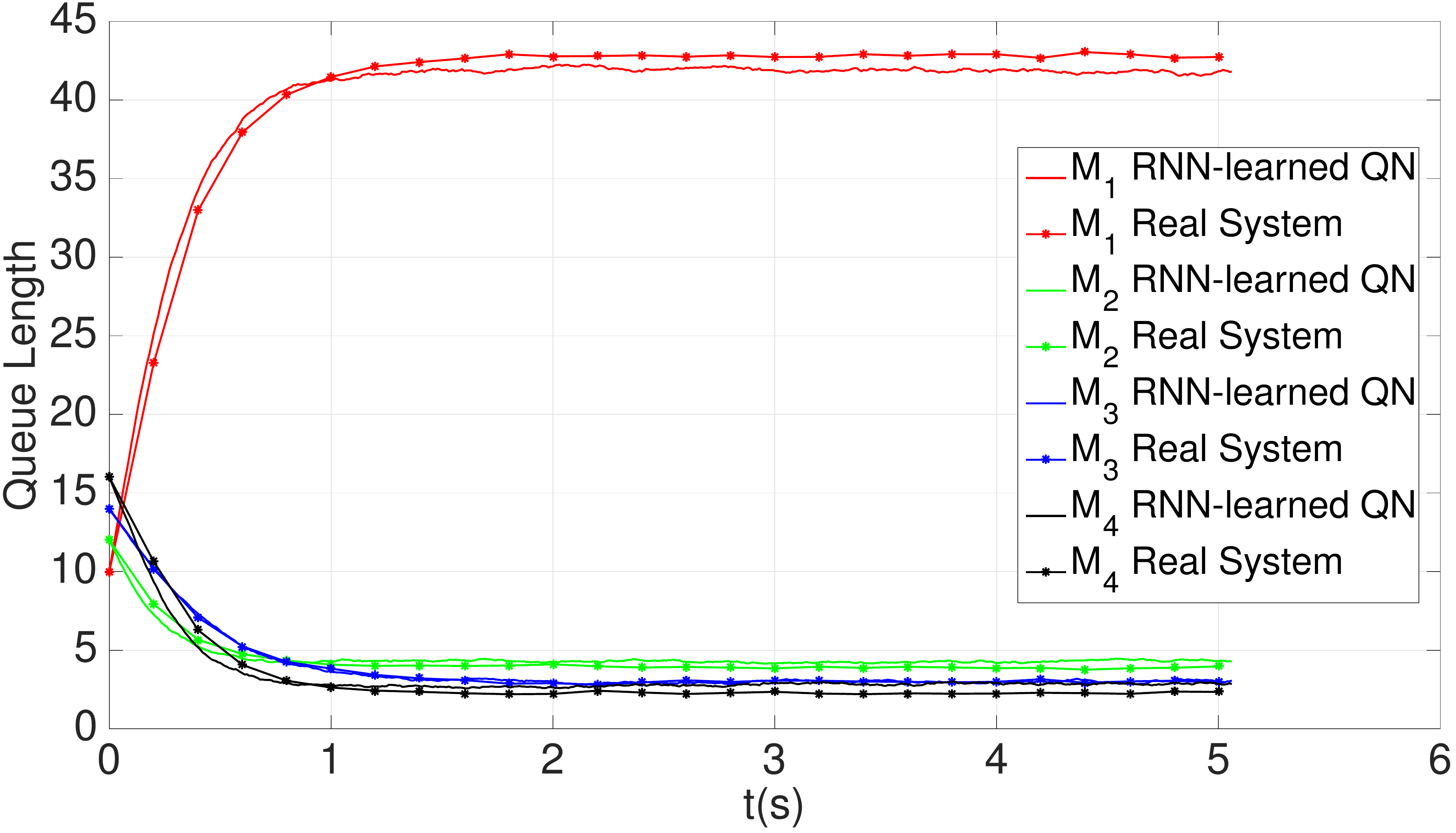}
		\caption{$k=2$, err=6.46\%}
		\label{fig:bottleneck2}
	\end{subfigure}
	\begin{subfigure}{0.246\linewidth}
		\centering
		\includegraphics[width=\linewidth]{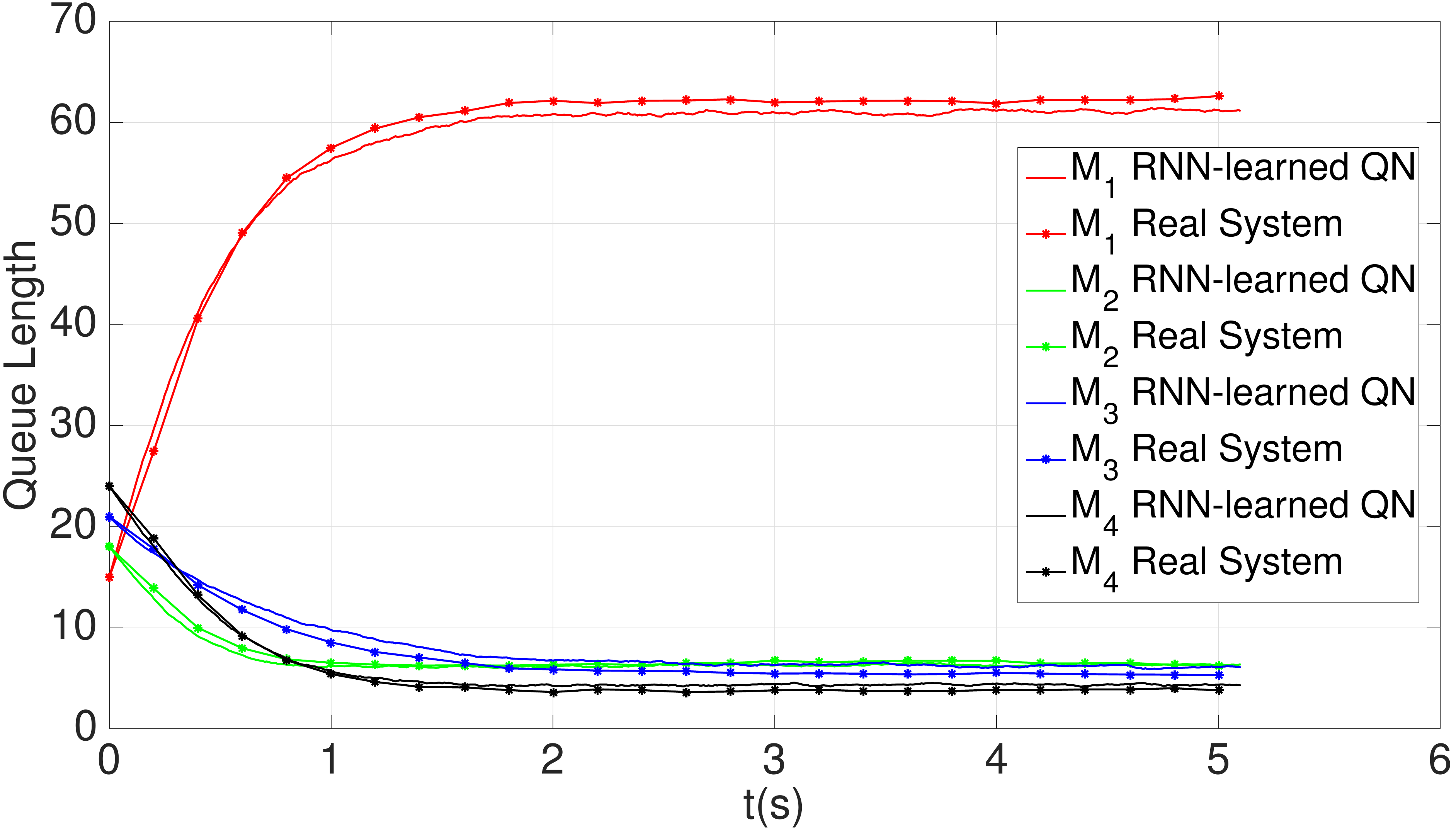}
		\caption{$k=3$, err=5.03\%}
		\label{fig:bottleneck3}
	\end{subfigure}
	\begin{subfigure}{0.246\linewidth}
		\centering
		\includegraphics[width=\linewidth]{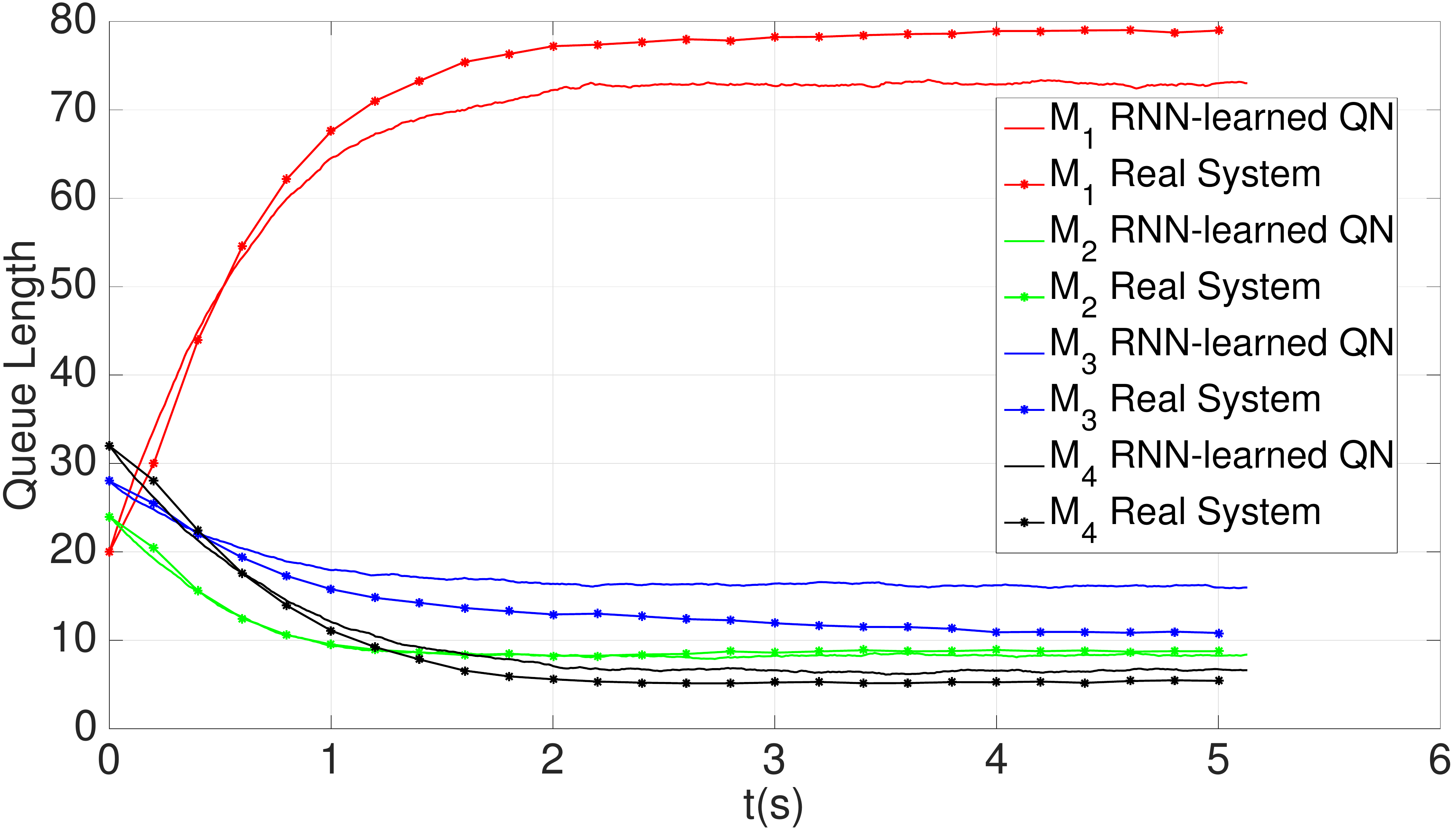}
		\caption{$k=4$, err=6.45\%}
		\label{fig:bottleneck4}
	\end{subfigure}
	\begin{subfigure}{0.246\linewidth}
		\centering
		\includegraphics[width=\linewidth]{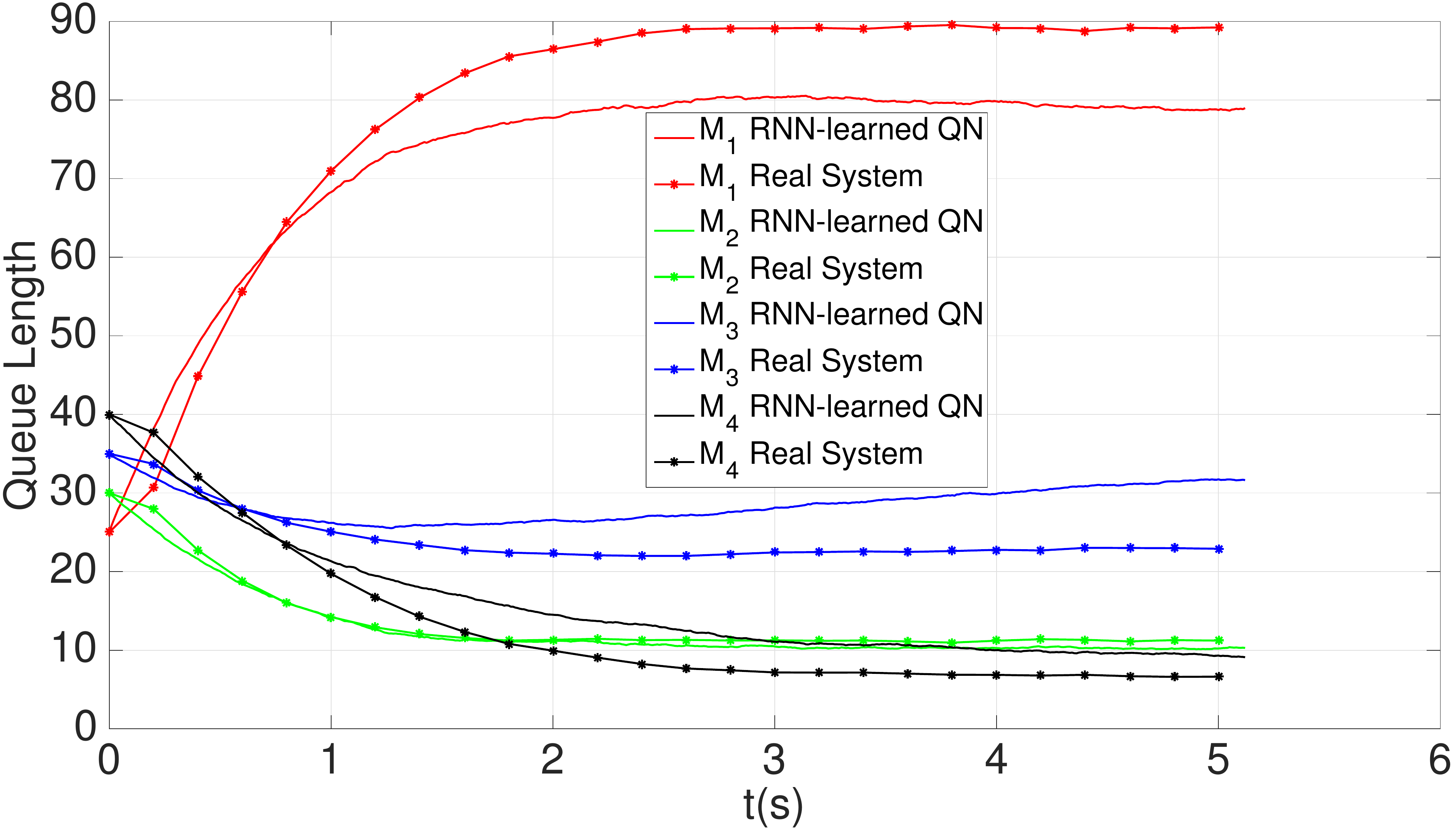}
		\caption{$k=5$, err=9.05\%}
		\label{fig:bottleneck5}
	\end{subfigure}
	\caption{Comparison between the real system dynamics (i.e., marked lines) and the RNN-learned QN (i.e., straight lines) in what-if cases over increasing circulating populations $N$, given by $N = 26k$.}
	\label{fig:bottleneck}
\end{figure*}

Specifically, we implemented \node{W} as a multi-threaded Python program. Each thread runs an independent concurrent user (i.e., one of the $N$ processes) that iteratively accesses the system, sleeping for an exponentially distributed delay between subsequent requests;  \node{LB} is a single-threaded NodeJs web server which act as a randomized load balancer. Finally, \node{C1}, \node{C2} and \node{C3} are multi-threaded NodeJs Clusters\footnote{\url{https://nodejs.org/api/cluster.html}} whose load is generated by sleeping for an exponential distributed delay (i.e., the average value is given as input parameter of each cluster). We remark that although we were able to roughly fix the distribution of the service demands their exact shape is still unknown since it is influenced by subtler factors that are hidden to developers (e.g., internal behavior of the web server, communication aspects). Moreover, in order to evaluate our learning methodology in an interesting scenario, we deployed the three replicas of the system with different parallelism levels and different service rates. 

Similarly to~\cite{wang2016maximum}, we collected the queue length traces used as input of the learning process (see Section~\ref{sec:learning}) by parsing the access logs generated by each component of the system. However, other monitoring solutions could be used, based for instance on recording the TCP backlog~\cite{ase17}. With this set-up, we were able to sample data with a measurement step $\Delta t = 0.01$\,s, which turned out to be sufficient for observing the transient dynamics of each component without altering the application behavior.  The replication package for this evaluation is publicly available at \url{https://zenodo.org/record/3679251}.

\emph{Model Learning:} We built the training dataset as a collection of queue length traces produced by the target application under 50 different initial population vectors where each station had a number of clients drawn uniformly at random between 0 and 30. For each such initial population vector the trace consisted of the average queue-length dynamics over 500 independent executions.

The target model of the learning process is reported in the right side of Figure~\ref{fig:casestudy}. In particular, components \node{C1}, \node{C2}, \node{C3} are modeled by queuing stations $M_2$, $M_3$, $M_4$, while both the workload generator \node{W} and the load balancer \node{LB} are abstracted by the same station $M_1$, since the delay introduced by \node{LB} is negligible with respect to the other components of the network. All the parameters of the resulting QN were considered parameters to be learned by the RNN. 

Similarly to the synthetic case, the collected traces were split into two halves for training and validation, respectively. We used Adam~\cite{Adam} as the learning algorithm with learning rate equal to $0.01$ and iterated until the error computed on the validation set did not improve at least $0.01\%$ in the last 50 iterations. With this, the system parameters were learned in 27 minutes on average, with a validation error of 3.89\%.

\emph{What if analysis:} In the following we evaluate the predictive power of the RNN learned QN under an unseen number of clients, concurrency levels and routing probabilities. Differently from the synthetic case, here we emulate a concrete usage scenario in which an initially hidden performance bottleneck is discovered and solved only relying on the insights given by the learned model. For doing so, we exercised both the QN model and the real system under an increasing number of clients (here each simulation averaged over 300 simulation runs instead of 500 since for evaluating the what-if analysis less runs are needed) by a factor $k = 2, \ldots, 5$ with respect to an initial population which had 26 circulating clients. Figure~\ref{fig:bottleneck} reports the numerical results of this evaluation, showing a trend that induces a saturation condition in station $M_3$. Overall, the prediction error of the RNN is less than 10\% across all instances. 

In Figure~\ref{figure:solvebottleneck} we report two different strategies that can be used in order to remove the bottleneck: we reevaluated both the learned model and the real system starting from the case $k =4$ (see  Figure~\ref{fig:bottleneck4}), varying either the number of servers or the load-balancing weights/routing probabilities. Figure~\ref{figure:solvebottlenecka} shows the dynamics of the system when the number of servers of $M_3$ is increased from 5 to 8, Figure~\ref{figure:solvebottleneckb} reports the what-if scenario in which we change the load distribution strategy  from a uniform  probability distribution to one where stations $M_2$, $M_3$, and $M_4$ are targeted with probability $0.35$, $0.20$ and $0.45$, respectively. Consistently with the intuition, both what-if instances show a lighter pressure (i.e., smaller queue length) at $M_3$. Furthermore, both situations are well predicted by the RNN, yielding an accuracy error of ca. 6\% with respect to the real system dynamics.


\begin{figure*}
	\centering
	\begin{subfigure}{0.40\linewidth}
		\centering
		\includegraphics[width=\linewidth]{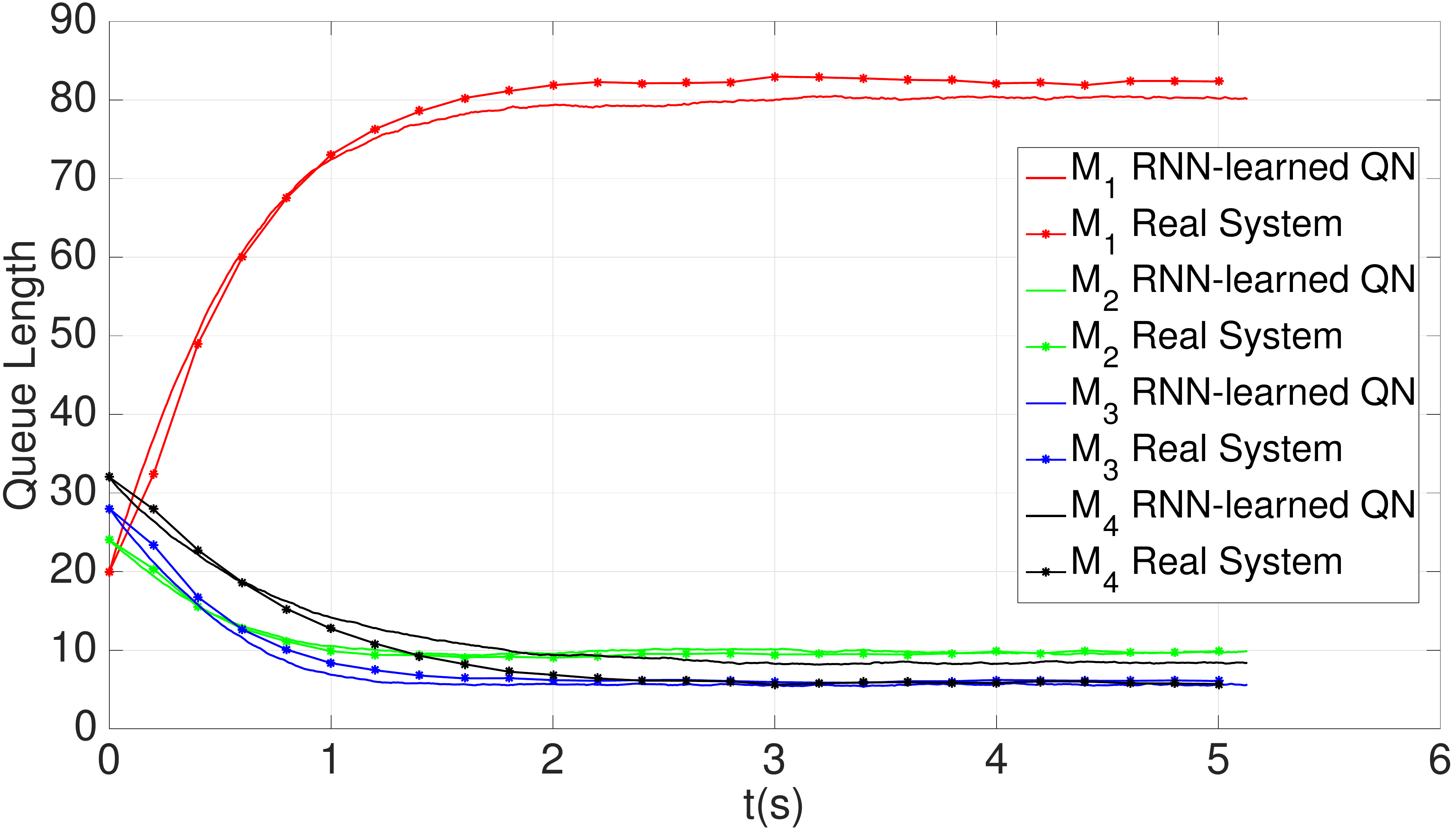}
		\caption{err=5.98\%}
		\label{figure:solvebottlenecka}
	\end{subfigure}
	\qquad \qquad 
	\begin{subfigure}{0.40\linewidth}
		\centering
		\includegraphics[width=\linewidth]{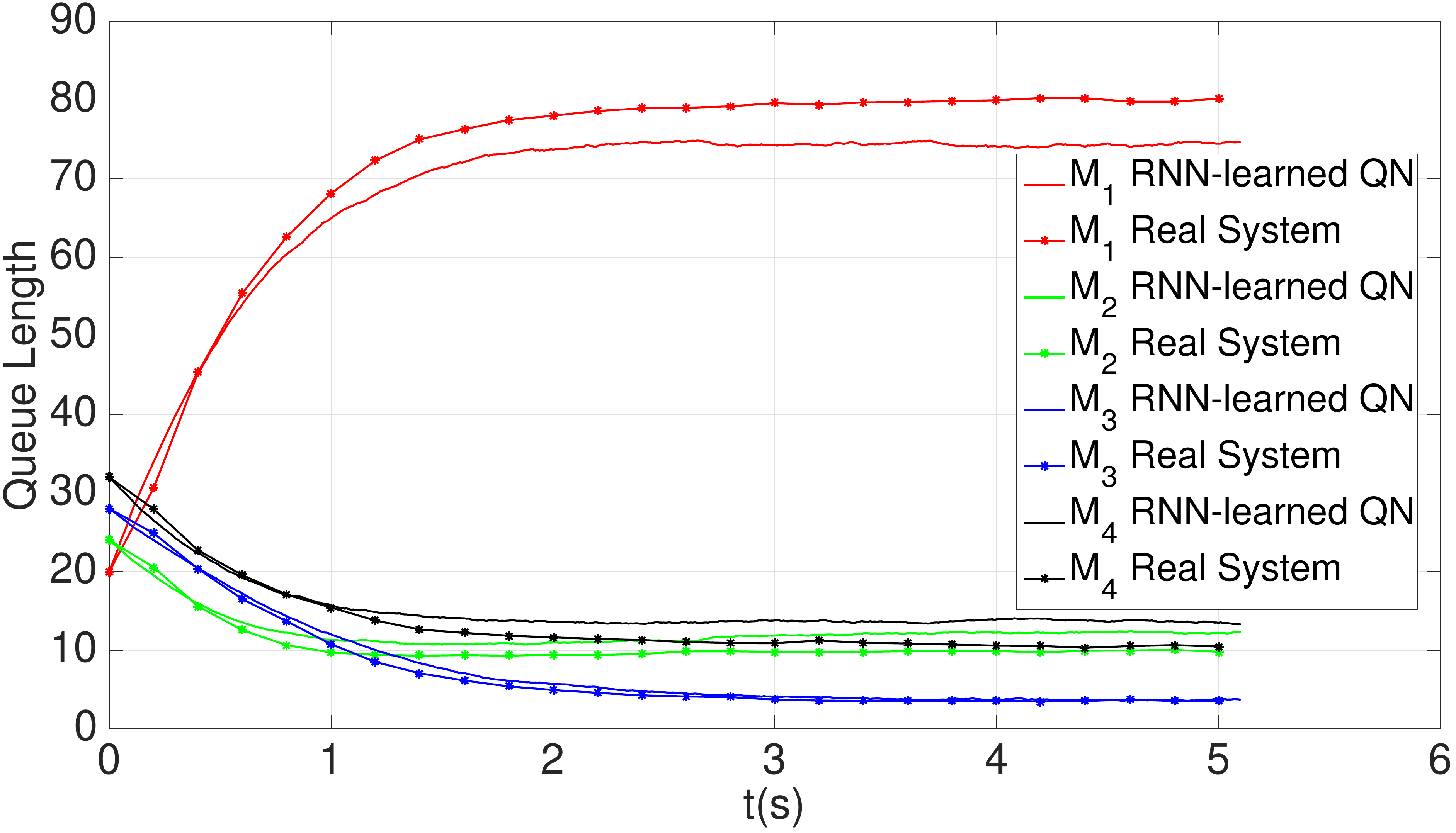}
		\caption{err=6.10\%}
		\label{figure:solvebottleneckb}
	\end{subfigure}
	\caption{a) What-if scenario changing the concurrency level of $M_3$ from 5 to 8. b) What-if scenario changing the load balancing strategy from a uniform probability distribution to the case where stations $M_2$, $M_3$, and $M_4$ are reached with probabilities $0.35$, $0.20$ and $0.45$, respectively. Both scenarios have been evaluated on the real case study and the RNN-learned QN with an initial population vector with $k = 4$ from  Figure~\ref{fig:bottleneck4}.}
	\label{figure:solvebottleneck}
\end{figure*}


\section{Related work}\label{sec:related}

In this section we relate against techniques related to the following lines of research: performance prediction from programs, generation of performance models from programs, and estimation of parameters in QNs. 

\subsection{Program-driven Performance Prediction}\label{sec:rel:pred}

A line of work focuses on the derivation of performance predictions from code analysis. \textsl{PerfPlotter} uses program analysis (specifically, probabilistic symbolic execution~\cite{Geldenhuys:2012:PSE:2338965.2336773}) to generate a \emph{performance distribution}, i.e., the probability distribution function of a performance metric such as response time~\cite{chen2016generating}. Thus, the result of the overall analysis is a quantitative model but it is not predictive. Furthermore, the approach applies to single-threaded applications, hence important performance-influencing sources such as threads contention cannot be captured. 

Other related approaches consist in predicting performance models using \emph{black-box} methods. They are particularly relevant for variability-intensive systems, where they relate configuration settings in a software system with their performance impact~\cite{Siegmund:2012:PPV:2337223.2337243,Siegmund:2015:PMH:2786805.2786845}. Machine-learning techniques have been used also in this case to build the predictive model~\cite{6693089,Siegmund:2015:PMH:2786805.2786845,DBLP:conf/wosp/ValovPGFC17,DBLP:conf/sigsoft/JamshidiVKS18}.
For instance, in~\cite{Siegmund:2012:PPV:2337223.2337243} the system model is assumed to be a linear combination of binary variables (e.g., tree structured models), each of them denoting the presence or the absence of a feature. Then the performance model is computed by means of linear regression over pairs of configurations and measured performance indices. The influence of possible feature interactions is embedded in the model by introducing fresh variables so as to preserve the linear structure of the model.  As discussed in~\cite{DBLP:conf/wosp/ValovPGFC17}, these black-box approaches can be seen as complementary to ours, which can provide a reliable mathematical abstraction by which performance can be explicitly associated to software components, thus increasing the explanatory power of the prediction. 

\subsection{Program-driven Generation of Performance Models}

While model-driven approaches to software performance have been  researched quite intensively~\cite{DBLP:books/daglib/0027475}, program-driven generation of performance models has been less explored, and has been concerned with specific kinds of applications. Indeed, the early approach by Hrischuk et al. is concerned with the generation of software performance models (specifically, layered queuing networks~\cite{10.1109/TSE.2008.74}) from a class of distributed applications whose components communicate solely by remote procedure calls~\cite{hrischuk1995automatic}. Brosig et al. derive a component-based performance model from applications running on the Java EE platform~\cite{Brosig:2009:AEP:1698822.1698835,brosig2011automated}. Tarvo and Reiss develop a technique for the extraction of discrete-event simulation models from a class of multi-threaded programs covering task-oriented applications, whereby the business logic consists in assigning a given workload (i.e., a task) to a number of worker threads from a pool~\cite{DBLP:conf/kbse/TarvoR14}. Their use of a simulation model as opposed to an analytical model is justified by the difficulty in building the latter, especially to model such diverse performance-related phenomena as queuing effects, inter-thread synchronization, and hardware contention. This is indeed the limitation that we aim to overcome with our approach, by building the analytical model automatically from measurements.

\subsection{Estimation of service demands in queuing networks}\label{sec:rel:qn}

Most of the literature concerning the estimation of QN parameters focuses on service demands. In particular, it considers the situation when the system is in the steady-state, i.e., when a sufficiently large amount of time has passed such that its behavior does not depend on the initial conditions~\cite{spinner2015evaluating}. Mathematically, the assumption of a steady-state regime enables the leveraging of a wealth of analytical results for QNs~\cite{bolch}. Based on these are several estimation methods using techniques such as linear regression~\cite{pacifici2008cpu}, quadratic programming~\cite{mascots18}, non-linear optimization~\cite{menasce2008computing,menasceIcpe17}, clustering regression~\cite{cremonesi2010service}, independent component analysis~\cite{sharma2008automatic}, pattern matching~\cite{cremonesi2014indirect}, Gibbs sampling~\cite{wang2013bayesian, sutton2011bayesian}, and maximum likelihood~\cite{wang2016maximum}. 

The main advancement of our approach with respect to the state of the art is the ability to learn the whole model, i.e., both the service demands and the QN topology (via the routing probabilities). In addition, since it uses an ODE representation, it does not make assumptions about the stationarity of the system; indeed, we do train our RNN using traces that include the transient dynamics. Actually, our approach uses the same QN model as the service-demand estimation method recently proposed in~\cite{mascots18}, which is also based on fluid approximation. 

Another difference with practical implications regards the type of data using for the estimation.  Approaches such as~\cite{liu2006parameter,cremonesi2010service,sharma2008automatic,kalbasi2011mode,cremonesi2014indirect} require measurements of quantities that may be difficult to obtain. For example, utilization metrics may not be available to the user when there is no complete information about the underlying hardware stack, for instance in a virtualized system running on a Platform-as-a-Service environment. Instead, measuring queue-length samples only has been regarded as more advantageous~\cite{wang2016maximum,mascots18}, since this information can be often obtained from application logs or by means of operating system calls.

\section{Conclusions}\label{sec:conclusion}

We presented a novel methodology for learning queuing network (QN) models of software systems. The main novelty lies in the encoding of the QN as an explainable recurrent neural network where inputs and weights are associated to standard queuing network inputs and parameters. 
We reported promising results on synthetic examples and on a real case study, where the maximum discrepancy between the dynamics predicted by the learned models and those computed through the ground truth is less than the $10$\% when the system is evaluated under unseen configurations  that are not included in the training set. We plan to extend our technique for capturing more complex models and systems, such as mixed multi-class and layered QNs, and to explore other learning methodologies such as neural ODEs~\cite{chen2018neural} and residual networks~\cite{zagoruyko2016wide}. Moreover, in order to improve the accuracy of the learned models and to reduce the simulation time, we plan to investigate active learning techniques that enable an informed sampling of the initial conditions~\cite{kaltenecker2019distance}.

\bigskip
\acks{This work has been partially supported by the PRIN project ``SEDUCE'' no. 2017TWRCNB.}


\appendix
\section{Appendix}
\label{app:not_univoque_proof}
\begin{proof}[Proof of Theorem \ref{not_univoque}]
We construct $\hat{\Pt}$ and $\hat{\rate}$ as follows:
\begin{align*}
	\hat{\Pt}_{k,i} & = 
	\begin{cases}
		\pi_k &\mbox{ if $i = k$}\\
		\frac{\Pt_{k,i}}{1-\Pt_{k,k}}(1-\pi_k) &\mbox{ if $\Pt_{k,k} < 1$ and $i≠k$}\\
		\frac{1-\pi_k}{M-1} &\mbox{ otherwise}
	\end{cases} \\
	\hat{\rate}_k & = 
	\begin{cases}
		\frac{\Pt_{k,k}-1}{\pi_k-1}\rate_k &\mbox{ if $\Pt_{k,k} < 1$}\\
		0 &\mbox{ otherwise}
	\end{cases}
\end{align*}

We prove that, for each $i≠k$ we have $\hat{\Pt}_{k,i}\hat{\rate}_k = \Pt_{k,i}\rate_k$ and $(\hat{\Pt}_{k,k}-1)\hat{\rate}_k = (\Pt_{k,k}-1)\rate_k$. Then (a) follows by substitution.

We now consider the case $\Pt_{k,k} < 1$.
\begin{align*}
\hat{\Pt}_{k,i}\hat{\rate}_k &= \frac{\Pt_{k,i}}{1-\Pt_{k,k}}(1-\pi_k) \frac{\Pt_{k,k}-1}{\pi_k-1}\rate_k\\
&= \frac{\Pt_{k,i}}{\Pt_{k,k}-1}(\pi_k-1) \frac{\Pt_{k,k}-1}{\pi_k-1}\rate_k\\
&= \Pt_{k,i}\rate_k.\\
(\hat{\Pt}_{k,k}-1)\hat{\rate}_k &= (\pi_k-1) \frac{\Pt_{k,k}-1}{\pi_k-1}\rate_k\\
&= (\Pt_{k,k}-1)\rate_k.
\end{align*}
We now consider the case $\Pt_{k,k} = 1$. We remark that, in this case, $\Pt_{k,i} = 0$ if $i≠k$.
\begin{align*}
\hat{\Pt}_{k,i}\hat{\rate}_k 
& = \frac{1-\pi_k}{M-1} 0 = 0 = 0 \rate_k = \Pt_{k,i} \rate_k . \\
(\hat{\Pt}_{k,k}-1)\hat{\rate}_k &= (\pi_k-1) 0 = 0 = 0 \rate_k = (\Pt_{k,k}-1) \rate_k .
\end{align*}

The point (b) is true by definition of $\hat{\Pt}$. Statement (c) can be shown as follows. When $\Pt_{k,k} < 1$:
\begin{align*}
\sum_i \hat{\Pt}_{k,i} &= \hat{\Pt}_{k,k} + \sum_{i≠k} \hat{\Pt}_{k,i}\\
&= \pi_k + \sum_{i≠k}\frac{\Pt_{k,i}}{1-\Pt_{k,k}}(1-\pi_k)\\
&= \pi_k + 1-\pi_k = 1
\end{align*}
where the last statement follows because $\sum_i \Pt_{k,i}=1$, $\sum_{i \neq k} \Pt_{k,i} = 1-\Pt_{k,k}$.
When $\Pt_{k,k} = 1$:
\begin{align*}
\sum_i \hat{\Pt}_{k,i} &= \hat{\Pt}_{k,k} + \sum_{i≠k} \hat{\Pt}_{k,i}\\
&= \pi_k + \sum_{i≠k}\frac{1-\pi_k}{M-1}\\
&= \pi_k + \frac{M-1}{M-1}(1-\pi_k)\\
&= \pi_k + 1-\pi_k = 1
\end{align*}
Statement (d) can be shown observing that $0 ≤ \pi_k < 1$, $1-\Pt_{k,k} ≥ 0$ (since $\Pt_{k,k}≤1$) and $1-\pi_k > 0$.
Statement (e) can be shown observing that $\rate_k ≥ 0$, $\Pt_{k,k}-1 ≤ 0$ and $\pi_k-1 < 0$.
\end{proof}

\end{document}